\documentclass[journal]{IEEEtran}

\usepackage{amsmath,amssymb,bm,mathtools}
\usepackage{cite}
\usepackage{graphicx}
\usepackage{microtype}
\usepackage{amsthm}
\usepackage{booktabs}
\usepackage{array}
\usepackage{algorithm}
\usepackage{algorithmic}
\usepackage[caption=false,font=footnotesize,labelformat=simple]{subfig}

\usepackage{url}
\usepackage{balance}
\usepackage{color}
\newtheorem{theorem}{Theorem}
\newtheorem{proposition}{Proposition}

\newtheorem{lemma}{Lemma}
\newtheorem{assumption}{Assumption}
\newtheorem{remark}{Remark}

\usepackage[hidelinks]{hyperref}

\renewcommand{\bm}[1]{\mathbf{#1}}         
\newcommand{\bmg}[1]{\boldsymbol{#1}}      
\newcommand{\jmathc}{\mathrm{j}}
\newcommand{\diag}{\operatorname{diag}}
\newcommand{\rank}{\operatorname{rank}}
\newcommand{\tr}{\operatorname{tr}}
\newcommand{\range}{\operatorname{range}}
\newcommand{\nulls}{\operatorname{null}}
\newcommand{\sinc}{\operatorname{sinc}}
\newcommand{\norm}[1]{\left\lVert #1\right\rVert}
\newcommand{\abs}[1]{\left\lvert #1\right\rvert}

\newcommand{\C}{\mathbb{C}}
\newcommand{\R}{\mathbb{R}}
\newcommand{\Z}{\mathbb{Z}}

\newcommand{\one}{\mathbf 1}
\newcommand{\Wwrap}{\mathcal W_{2\pi}}
\newcommand{\Qphi}{\mathcal Q_{B_\phi}}
\newcommand{\st}{\mathrm{st}}

\graphicspath{{figs/}{./}}

\begin{document}
\bstctlcite{IEEEexample:BSTcontrol}

\title{Space--Time Stacked Intelligent Metasurfaces for Doppler-Resilient LEO Broadcasting}

\author{Chandan Kumar Sheemar,~\IEEEmembership{Member,~IEEE},
Giovanni Iacovelli,~\IEEEmembership{Member,~IEEE},
Sourabh Solanki,~\IEEEmembership{Member,~IEEE},\\ and
Symeon Chatzinotas,~\IEEEmembership{Fellow,~IEEE}%
\thanks{C. K. Sheemar, G. Iacovelli, and S. Chatzinotas are with the Interdisciplinary Centre for Security, Reliability and Trust (SnT), University of Luxembourg, L-1855 Luxembourg (e-mail: \{chandankumar.sheemar, giovanni.iacovelli, symeon.chatzinotas\}@uni.lu). S. Solanki is with the National Institute of Technology Warangal, India (e-mail: ssolanki@nitw.ac.in).}
}

\maketitle
 
 \begin{abstract}
Stacked intelligent metasurfaces (SIMs) cascade programmable layers in front of a feed. In the SIM designs reported so far, including those proposed for low-Earth-orbit (LEO) satellites, the layer coefficients are space-only controls or are varied randomly in time to support opportunistic scheduling, so that every radiated beam carries the frequency of the feed. This paper introduces space--time SIMs (ST-SIMs), whose layers are modulated in time, and uses the temporal dimension of the stack deterministically: each spatial mode of the aperture is paired with its own prescribed frequency translation. The application is LEO broadcasting, where a common waveform reaches terminals with different Doppler shifts, and any common Doppler correction leaves residuals of up to half the Doppler span. We derive the harmonic-domain model of the stack, solve the minimum-energy space--frequency synthesis in closed form, and prove that full column rank of the steering matrix is necessary and sufficient for decoupling $K$ terminal-specific translations, a result that extends to groups of terminals. We further propose a phase-only synthesis algorithm with guaranteed stationarity, together with quantization and control-rate laws. Simulations at $20$~GHz over $50$ random drops show that a four-layer ST-SIM with a single radio-frequency chain removes the Doppler floor of a static SIM for footprint radii above about $45$~km and, under the ideal phase-only model and at equal radiated peak energy, approaches ideal digital per-terminal precompensation within $0.5$~dB for resolvable terminals, whereas, in the evaluated geometries, a single time-modulated layer cannot synthesize the required modes.
\end{abstract}


\section{Introduction}
\IEEEPARstart{L}{ow}-Earth-orbit (LEO) constellations are a cornerstone of fifth- and sixth-generation non-terrestrial networks (NTNs) owing to their global coverage, moderate latency, and standardized integration with terrestrial radio access \cite{Kodheli2021,kumar20235g,sheemarbeam}. Their orbital speed of about $7.5$~km/s is also their most persistent impairment, since the satellite--terminal range changes by kilometers per second and the Doppler shift scales with the carrier frequency. At $20$~GHz, the Doppler shift exceeds $\pm400$~kHz, i.e., more than ten subcarrier spacings of a 30-kHz orthogonal frequency-division multiplexing (OFDM) numerology, and it drifts at rates of several kHz/s \cite{Ali1998,You2020,3GPP38821}. An uncompensated carrier-frequency offset (CFO) of this size destroys subcarrier orthogonality and generates inter-carrier interference (ICI) \cite{Pollet1995,Moose1994}, and Doppler handling permeates every NTN procedure \cite{Kodheli2021RA,Lin2021Doppler,Meshram2025}. Since orbital motion is deterministic, the predominant remedy is precompensation, whereby the transmitter, i.e., the gateway for a transparent and the satellite for a regenerative payload, shifts the transmitted carrier by the Doppler predicted from the ephemeris \cite{3GPP38821,Lin2021Doppler,Kim2022}. Such a correction is exact for a single direction. A broadcast beam, however, serves terminals with different line-of-sight (LoS) directions and radial velocities, and, as shown in Section~\ref{sec:ideal}, the best common translation leaves a worst-terminal residual equal to half the Doppler span, a floor independent of the transmit power, the noise level, and any digital processing of the common waveform. For a 200-km footprint radius at $20$~GHz this span exceeds $300$~kHz, i.e., ten subcarrier spacings, and the residual is left to the terminals \cite{Meshram2025,Kodheli2021RA}.

  Stacked intelligent metasurfaces (SIMs) cascade $L$ programmable metasurface layers of $N$ meta-atoms each in front of a feed, so that inter-layer diffraction implements a high-dimensional transformation of the incident wave \cite{AnJSAC2023,AnICC2023,iacovelli2026stacked}. Each layer applies a diagonal phase profile to the field impinging on it, and the free-space propagation between consecutive layers mixes the outputs of all meta-atoms, so that the cascade realizes a dense $N\times N$ linear operator with only $LN$ phase controls and without any digital multiplication. The stack can thus perform beamforming, precoding, or a discrete Fourier transform in the wave domain at the speed of light, which reduces the number of RF chains, the baseband load, and the power consumption, while its depth $L$ and its $N$ meta-atoms per layer provide degrees of freedom that a single-layer surface lacks. SIMs have been studied for multiuser beamforming \cite{An2025}, channel estimation \cite{Yao2024}, physically consistent modeling \cite{Nerini2024}, LEO satellite downlinks \cite{Lin2024,sheemar2026one}, wideband MIMO-OFDM \cite{Li2026}, unmanned aerial vehicles (UAVs) \cite{sheemar2026energy,sheemarUAV2026stacked}, and doubly-dispersive channels \cite{Ranasinghe2026,SheemarSurvey2026}. In all of these works, the layer coefficients are space-only control variables: the stack shapes the spatial or frequency-selective response of a beam, but every beam carries the frequency of the feed, and temporal modulation of the layers has not been exploited as a deterministic signal-processing degree of freedom.

\subsection{State-of-the-Art}
Doppler compensation in NTNs has been studied extensively. The 3GPP NTN study \cite{3GPP38821} relies on ephemeris-assisted precompensation at the network and on CFO estimation from reference signals at the terminal \cite{Lin2021Doppler}, and dedicated solutions have been proposed for random access \cite{Kodheli2021RA}, narrowband Internet-of-Things (IoT) uplinks \cite{Kim2022}, multibeam initial synchronization \cite{Meshram2025}, and massive multiple-input multiple-output (MIMO) LEO downlinks with per-user Doppler and delay precompensation in the digital domain \cite{You2020}. Doppler-robust waveforms \cite{Hadani2017,Bemani2023} relocate the problem to the receiver equalizer, and reconfigurable intelligent surfaces (RISs) have mitigated Doppler in high-mobility links \cite{deSouza2024} as externally illuminated reflectors rather than as the satellite front end.

In parallel, the electromagnetics literature has established that time-varying surfaces manipulate frequency. Time-modulated arrays were used decades ago to shape sidelobes and generate harmonics \cite{Kummer1963}, and a single-layer time-modulated array has been shown to separate downlink users by pre-shifted carriers with one RF chain \cite{He2015}. Space--time gradient metasurfaces \cite{Hadad2015} and spacetime metamaterials \cite{Caloz2020} showed that joint spatial and temporal variation yields frequency conversion and nonreciprocity, and space--time-coding digital metasurfaces \cite{Zhang2018} demonstrated programmable control of the propagation direction and harmonic power distribution, independent control of harmonic amplitude and phase \cite{Dai2018}, joint multi-frequency beam shaping and steering \cite{Castaldi2021}, Doppler-effect emulation and frequency conversion \cite{Ramaccia2020,Taravati2021,Taravati2022}, space- and frequency-division multiplexing \cite{Zhang2021}, sideband-free frequency conversion in a waveguide-fed antenna \cite{Wu2022}, and frequency conversion with beam shaping on a 1-bit metasurface antenna \cite{Sun2025,Gholami2026}. Closest to the present work, a single-layer 1-bit spatiotemporal reflecting metasurface on the propagation path offsets the Doppler shift of one high-mobility link by steering one frequency-translated harmonic toward the user, with a switching sequence obtained by a neural network \cite{Zhu2024}. Stacks of spatially uniform time-varying metasurfaces realize nonreciprocal isolation and frequency circulation \cite{Movahediqomi2026}, and in the SIM literature temporal variation has so far served only to randomize the channel for opportunistic scheduling \cite{Darsena2026}.

\begin{table}[t]
\centering
\caption{Positioning With Respect to Time-Varying Metasurfaces and SIMs}
\label{tab:position}
\scriptsize
\setlength{\tabcolsep}{2pt}
\renewcommand{\arraystretch}{1.1}
\begin{tabular}{@{}l c c c c c c@{}}
\toprule
Work & \shortstack{Multi-\\layer} & \shortstack{Time-\\varying} & \shortstack{Prescribed\\shift} & \shortstack{Distinct\\per beam} & \shortstack{Doppler-\\matched} & \shortstack{Stack\\laws}\\
\midrule
Space-only SIM \cite{AnJSAC2023,An2025,Lin2024} & \checkmark & -- & -- & -- & -- & --\\
Randomized ST-SIM \cite{Darsena2026} & \checkmark & \checkmark & -- & -- & -- & --\\
Stacked TV surfaces \cite{Movahediqomi2026} & \checkmark & \checkmark & \checkmark & -- & -- & --\\
Time-modulated array \cite{He2015} & -- & \checkmark & \checkmark & \checkmark & -- & --\\
STC metasurfaces \cite{Zhang2018,Castaldi2021,Wu2022} & -- & \checkmark & \checkmark & \checkmark & -- & --\\
Doppler-offset STM \cite{Zhu2024} & -- & \checkmark & \checkmark & -- & 1 link & --\\
This work & \checkmark & \checkmark & \checkmark & \checkmark & \checkmark & \checkmark\\
\bottomrule
\end{tabular}
\end{table}

Table~\ref{tab:position} positions this paper. In single-layer time-modulated surfaces, the beam frequencies are harmonics that serve as channels \cite{He2015,Castaldi2021,Zhang2021,Wu2022}, or a single translation offsets the Doppler of one link \cite{Zhu2024}, so that terminals with different Doppler shifts retain half their Doppler span (Proposition~\ref{prop:common}); in multilayer stacks, the time dimension has served randomization \cite{Darsena2026} or, with uniform layers, nonreciprocity \cite{Movahediqomi2026}. The deterministic pairing of distinct spatial modes of a stack with Doppler-matched translations, and the conditions under which a finite-depth, finite-resolution, and finitely fast phase-only stack realizes it, have not been investigated.  

\subsection{Contributions}
This paper proposes space-time (ST-) SIMs for deterministic space--frequency processing at the LEO satellite (Fig.~\ref{fig:concept}). The spatial modes and the temporal translations cannot be designed independently, since every frequency-translated mode must be spatially nulled at the unintended terminals and the resulting field must be synthesized by a finite-depth, finite-resolution, and finitely fast phase-only stack. The main contributions of our work are:
 \begin{itemize}
\item We present the \emph{ST-SIM architecture} and derive its \emph{harmonic-domain model}, which shows that stacking convolves the temporal spectra of the layers while mixing their spatial fields, and that desired modes and unwanted sidebands share one passive energy budget.
\item We develop the \emph{space--frequency synthesis theory}: we solve the minimum-energy multiuser synthesis in closed form with its complete solution set, geometry-dependent energy bounds, and long-run peak energy; we prove that full column rank of the steering matrix is necessary and sufficient for exactly decoupling $K$ terminal-specific translations, a condition that grouped compensation relaxes; we show that any common translation, hence any static SIM or single-translation time-modulated surface, leaves half the Doppler span; and we quantify in closed form the OFDM interference caused by residual Doppler and ephemeris errors.
\item We propose a \emph{realizable design}: a phase-only synthesis algorithm with closed-form gain scaling, reusable multilayer gradients, guaranteed monotone decrease to stationary points, and quantization-aware refinement, and we derive explicit quantization and control-rate laws.
\item We demonstrate the architecture in a \emph{link-level study} at $20$~GHz over $50$ random drops, in which we compare the ST-SIM with a single-layer time-modulated aperture, a static SIM with optimal common precompensation, with the same payload combined with per-terminal CFO correction, and with ideal digital per-terminal precompensation, and we map the operating regimes in footprint radius, aperture size, number of terminals, and grouping.
\end{itemize}

\begin{figure}[t]
\centering
\includegraphics[width=\columnwidth]{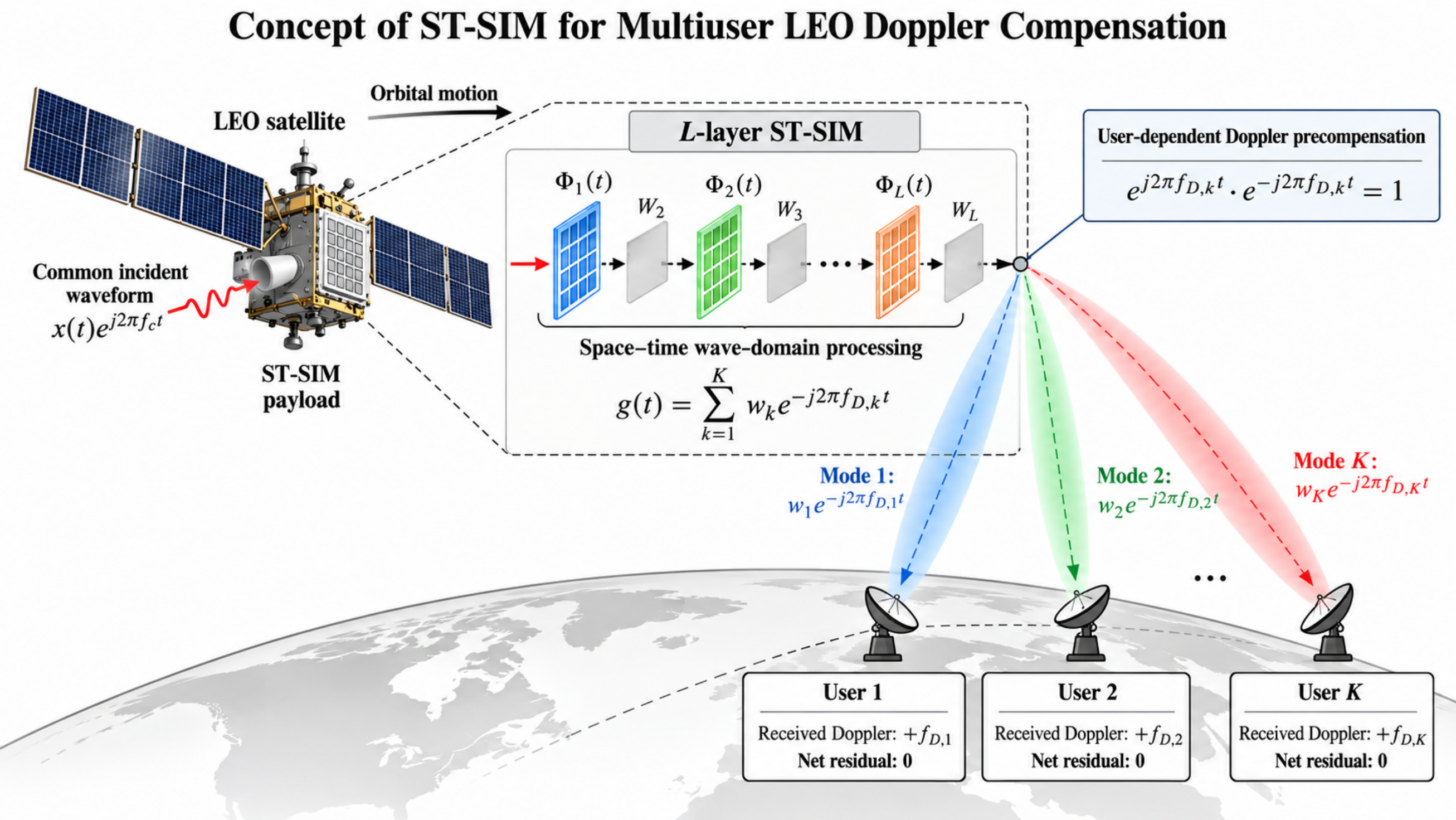}
\caption{User-dependent Doppler precompensation in the electromagnetic domain with an ST-SIM. A single feed illuminates an $L$-layer stack of time-varying phase-only metasurfaces on a LEO satellite, which transforms the common waveform into $K$ spatial modes directed toward the $K$ terminals; mode $k$ is translated in frequency by $-f_{D,k}$ so that terminal $k$ receives a carrier-Doppler-free signal.}
\label{fig:concept}
\end{figure}

\emph{Paper Organization:} Section~\ref{sec:model} presents the system model, Section~\ref{sec:ideal} the space--frequency synthesis theory, and Section~\ref{sec:practical} the phase-only synthesis algorithm and its hardware laws. Section~\ref{sec:results} reports the numerical results, and Section~\ref{sec:conclusion} concludes the paper.

\emph{Notation:} Boldface lowercase and uppercase letters denote vectors and matrices. $(\cdot)^T$, $(\cdot)^H$, $(\cdot)^*$, and $(\cdot)^\dagger$ denote transpose, Hermitian transpose, conjugation, and Moore--Penrose pseudoinverse; $\norm{\cdot}_2$ and $\norm{\cdot}_F$ the Euclidean/spectral and Frobenius norms; $\lambda_{\min}(\cdot)$, $\lambda_{\max}(\cdot)$ the extreme eigenvalues of a Hermitian matrix; $\diag(\bm x)$ the diagonal matrix built from $\bm x$; $\bm I_N$ the $N\times N$ identity; $\bm e_n$ the $n$th canonical vector; $\Wwrap(\cdot)$ the entrywise phase wrap to $[-\pi,\pi)$; and $\sinc(x)=\sin(\pi x)/(\pi x)$.
\section{System Model}\label{sec:model}
We consider a 6G multicast--broadcast service in which one LEO beam delivers a common data stream to $K$ registered terminals whose positions are known to the network from their registration reports. 

\subsection{LEO Propagation and Doppler}
Let $u$ denote transmit time and $t$ receive time. In an Earth-centered inertial frame, the satellite and terminal $k$ have positions $\bm r_s(\cdot)$ and $\bm r_k(\cdot)$, a wavefront emitted at $u$ is received at $t=u+\tau_k(u)$, and the one-way light time $\tau_k(u)$ solves the light-time equation $c\,\tau_k(u)=\norm{\bm r_k(u+\tau_k(u))-\bm r_s(u)}_2$ \cite{montenbruck2002satellite}. Let $x(u)$ be the unit-power complex envelope applied to a common RF feed at carrier $f_c$, $\ell_k(u)\ge0$ the slowly varying amplitude gain, $\bm a_k(u)\in\C^N$ the LoS steering vector of the radiating aperture, and $\bm g(u)\in\C^N$ the aperture field produced by the ST-SIM. After ideal downconversion at $f_c$, the received baseband signal is
\begin{equation}
y_k(t)=\ell_k(u)\,\bm a_k^H(u)\bm g(u)\,x(u)\,e^{-\jmathc2\pi f_c\tau_k(u)}+n_k(t),
\label{eq:exactrx}
\end{equation}
with $u$ the retarded transmit time and $n_k$ additive noise. Model~\eqref{eq:exactrx} does not introduce Doppler as a separate impairment: the carrier Doppler is the time derivative of the delay-induced carrier phase. 
Appendix~\ref{app:lighttime} shows that, with $\bm u_k$ the unit line-of-sight (LoS) vector and $\bm v_k^{\rm rel}=\dot{\bm r}_k-\dot{\bm r}_s$ the relative velocity, the instantaneous carrier Doppler is $f_{D,k}=-(f_c/c)\bm u_k^T\bm v_k^{\rm rel}$ up to a relative correction below $2.5\times10^{-5}$ for LEO motion, Over one OFDM block of duration $T_b$ centered at $t_0$, the carrier phase is linear in $t$ up to $\pi\abs{\dot f_{D,k}}T_b^2/4$, the delay varies by at most $\norm{\bm v_k^{\rm rel}}_2T_b/c$, and the LoS direction rotates by at most $\norm{\bm v_k^{\rm rel}}_2T_b/d_k$, with $d_k$ the range\footnote{For the pass of Section~\ref{sec:results}, $\abs{\dot f_{D,k}}\le5.8$~kHz/s and $T_b\approx35.7~\mu$s give a phase error below $10^{-5}$~rad, a delay variation of $0.9$~ns (3\% of a sample), and a rotation of $4.5\times10^{-7}$~rad.}. Under these conditions, \eqref{eq:exactrx} reduces to the blockwise model
\begin{equation}
y_k(t)=\alpha_k e^{\jmathc2\pi f_{D,k}(t-t_0)}\bm a_k^H\bm g(t-\tau_k)x(t-\tau_k)+n_k(t),
\label{eq:blockrx}
\end{equation}
where $\alpha_k$ absorbs the constant phase and path gain. Model~\eqref{eq:blockrx} is a LoS model; multipath Doppler, phase noise, and oscillator offsets remain for the terminal's synchronization. The ST-SIM thus performs the coarse synchronization by removing the orbital Doppler of the LoS path, and the terminal performs only the fine synchronization.

\subsection{Multilayer ST-SIM Transfer Function}
The satellite carries an $L$-layer ST-SIM with $N$ meta-atoms per layer, illuminated by one RF feed, which radiates the OFDM signal at carrier $f_c$ with bandwidth $B=M\Delta f$,  through a fixed coupling vector $\bm f\in\C^N$, $\norm{\bm f}_2=1$. The time-varying phase matrix of layer $l$ is
\begin{equation}
\bm\Phi_l(u)=\diag\!\big(e^{\jmathc\theta_{l,1}(u)},\ldots,e^{\jmathc\theta_{l,N}(u)}\big),
\label{eq:Phil}
\end{equation}
and $\bm W_l\in\C^{N\times N}$, $l=2,\ldots,L$, is the fixed propagation operator from layer $l-1$ to layer $l$. The aperture field is
\begin{equation}
\bm g(u)=\bm G(u)\bm f,\qquad
\bm G(u)=\bm\Phi_L(u)\bm W_L\cdots\bm\Phi_2(u)\bm W_2\bm\Phi_1(u).
\label{eq:G}
\end{equation}
Let $\nu_{\max}=\sup_k\abs{f_{D,k}}$ and let $\tau_{\rm stack}$ be the transit time across the stack.
\begin{assumption}[Frequency-flat inter-layer propagation]\label{as:flat}
$\nu_{\max}/f_c\ll1$, so a common $\bm W_l$ applies to all generated sidebands.
\end{assumption}
\begin{assumption}[Negligible modulation during transit]\label{as:transit}
$\nu_{\max}\tau_{\rm stack}\ll1$, so all layers are evaluated at the same control time $u$ in \eqref{eq:G}.
\end{assumption}
For held control, the modulation spectrum of the layers extends to the update rate $f_u$ and its harmonics rather than to $\nu_{\max}$, so both conditions must hold with $\nu_{\max}$ replaced by the control bandwidth. For a passive stack we assume
\begin{equation}
\norm{\bm W_l}_2\le\kappa_l\le1,\qquad l=2,\ldots,L,
\label{eq:contractive}
\end{equation}
with equality for lossless propagation; the $\bm W_l$ are moreover assumed nonsingular, as the Rayleigh--Sommerfeld operators of Section~\ref{sec:results} are. Equations~\eqref{eq:Phil}--\eqref{eq:G} constitute the ideal phase-only wave-domain model of the SIM literature \cite{AnJSAC2023,An2025}: diagonal unit-modulus layers, fixed inter-layer operators, and no multiple reflections, state-dependent or amplitude--phase coupling, or bias-network parasitics. In this model, $\norm{\bm g(u)}_2^2$ is the energy of the element excitations, which equals the radiated power only up to the radiation matrix of the aperture; the accuracy of the single-pass cascade depends on the inter-layer spacing and on the element design \cite{Nerini2024,iacovelli2026stacked}, and the transmission efficiency $\eta_L$ is a property of the normalized model rather than of a device. 

\subsection{Harmonic-Domain Representation}
Suppose first that each layer is periodic with modulation period $T_m=1/f_m$, so that $\bm\Phi_l(u)=\sum_{p\in\Z}\bm C_l^{[p]}e^{\jmathc2\pi pf_mu}$ with diagonal Fourier coefficients $\bm C_l^{[p]}=T_m^{-1}\int_0^{T_m}\bm\Phi_l(u)e^{-\jmathc2\pi pf_mu}du$.
\begin{lemma}[Harmonic convolution across the stack]\label{lem:conv}
$\bm G(u)=\sum_{p\in\Z}\bm G^{[p]}e^{\jmathc2\pi pf_mu}$ with
\begin{equation}
\bm G^{[p]}=\sum_{p_1+\cdots+p_L=p}\bm C_L^{[p_L]}\bm W_L\cdots\bm C_2^{[p_2]}\bm W_2\bm C_1^{[p_1]},
\label{eq:harmonicconv}
\end{equation}
and the $p$th aperture harmonic is $\bm g^{[p]}=\bm G^{[p]}\bm f$.
\end{lemma}
\begin{IEEEproof}
Assume that the Fourier coefficients of every layer are absolutely summable, $\sum_p\norm{\bm C_l^{[p]}}_2<\infty$, which holds, e.g., for smooth phase trajectories or for finitely many harmonics, but not for a state held constant over a control slot, whose coefficients decay as $1/\abs p$; the held control is treated in Proposition~\ref{prop:zoh}, and the receiver filter retains finitely many of its harmonics, to which the lemma applies. Substituting $\bm\Phi_l(u)=\sum_{p_l\in\Z}\bm C_l^{[p_l]}e^{\jmathc2\pi p_lf_mu}$ into \eqref{eq:G} gives
\begin{equation}
\bm G(u)= \hspace{-3mm}\sum_{p_1,\ldots,p_L\in\Z} \hspace{-3mm}\bm C_L^{[p_L]}\bm W_L\cdots\bm C_2^{[p_2]}\bm W_2\bm C_1^{[p_1]}e^{\jmathc2\pi(p_1+\cdots+p_L)f_mu},
\end{equation}
where the order of the matrix factors is preserved because the scalar exponentials commute with all matrices, while the $\bm C_l^{[p_l]}$ and $\bm W_l$ do not commute in general. By submultiplicativity, the $L$-fold sum converges absolutely, with $\sum_{p_1,\ldots,p_L}\norm{\cdot}_2\le\prod_l\kappa_l\sum_{p_l}\norm{\bm C_l^{[p_l]}}_2<\infty$, so its terms can be regrouped freely. Collecting all tuples with $p_1+\cdots+p_L=p$ yields a series in $e^{\jmathc2\pi pf_mu}$ with coefficients \eqref{eq:harmonicconv}. Since $\bm G(u)$ is $T_m$-periodic and its Fourier coefficients are unique, these coefficients are the harmonics $\bm G^{[p]}$ of $\bm G(u)$. Finally, $\bm f$ is time-invariant, so $\bm g(u)=\bm G(u)\bm f=\sum_p\bm G^{[p]}\bm f\,e^{\jmathc2\pi pf_mu}$, i.e., $\bm g^{[p]}=\bm G^{[p]}\bm f$.
\end{IEEEproof}

Stacking thus convolves the temporal spectra of the layers while the inter-layer operators mix the associated spatial fields. The reconfiguration capabilities of the meta-atoms enter \eqref{eq:harmonicconv} only through the $\bm C_l^{[p]}$: a state held over a reconfiguration period $T_u$ weights the harmonics by $\sinc(pf_mT_u)$ and replicates them at multiples of $f_u=1/T_u$, so that only tones with $\abs{\nu}<f_u/2$ can be synthesized (Proposition~\ref{prop:zoh}), and a $B_\phi$-bit reconfiguration accuracy diverts energy into the unwanted harmonics $\bm e_{\st}(u)$ (Proposition~\ref{thm:quant}). Lemma~\ref{lem:conv} also specifies which modal fields the stack can produce. Let the $Q$ tones to be synthesized, $\nu_1,\ldots,\nu_Q$, lie on a grid of spacing $\delta_f$ (cf.\ Section~\ref{sec:practical}) and set $f_m=\delta_f$. Each tone is then a harmonic of the stack, $\nu_q=-p_qf_m$ with $p_q\in\Z$, and the aperture field decomposes as
\begin{equation}
\bm g(u)=\sum_{q=1}^{Q}\bmg\psi_qe^{-\jmathc2\pi\nu_qu}+\bm e_{\st}(u),\quad \bmg\psi_q=\bm g^{[p_q]},
\label{eq:modal}
\end{equation}
where $\bm e_{\st}(u)=\sum_{p\notin\{p_1,\ldots,p_Q\}}\bm g^{[p]}e^{\jmathc2\pi pf_mu}$ collects the unwanted harmonics generated by \eqref{eq:harmonicconv}. The number of modes $Q$ is the number of beams, each of which carries one frequency translation, whether it serves one terminal or several: $Q=K$ for one beam per terminal and $Q=G$ for one multicast beam per group in the grouped compensation. Synthesis therefore amounts to concentrating the harmonic content \eqref{eq:harmonicconv} on the indices $p_1,\ldots,p_Q$ with the prescribed spatial profiles; Section~\ref{sec:practical} performs this in the time domain. For tones that do not share a common grid, the control is quasi-periodic, and \eqref{eq:harmonicconv} holds with the integer indices replaced by the frequencies generated by the layers and index addition replaced by frequency addition.

 \begin{lemma}[Spectral-energy constraint]\label{lem:energy}
Let the inter-layer operators satisfy \eqref{eq:contractive} and let $\norm{\bm f}_2=1$. Then, for every control trajectory $\{\bm\theta_l(u)\}_{l=1}^L$, the aperture field \eqref{eq:G} satisfies
\begin{equation}
\norm{\bm g(u)}_2\le1,\qquad \forall u\in\R,
\label{eq:energyinst}
\end{equation}
with equality for all $u$ if $\bm W_l$ is unitary for every $l=2,\ldots,L$. Moreover, if the control is periodic with period $T_m$, the aperture harmonics $\bm g^{[p]}$ of Lemma~\ref{lem:conv} satisfy
\begin{equation}
\sum_{p\in\Z}\norm{\bm g^{[p]}}_2^2=\frac{1}{T_m}\int_0^{T_m}\norm{\bm g(u)}_2^2\,du\le1.
\label{eq:energyharm}
\end{equation}
\end{lemma}
 \begin{IEEEproof}
Fix $u$. Each phase matrix $\bm\Phi_l(u)$ in \eqref{eq:Phil} is diagonal with unit-modulus entries, so $\bm\Phi_l^H(u)\bm\Phi_l(u)=\bm I_N$, i.e., $\bm\Phi_l(u)$ is unitary and preserves the Euclidean norm of every vector, $\norm{\bm\Phi_l(u)\bm x}_2=\norm{\bm x}_2$; in particular $\norm{\bm\Phi_l(u)}_2=1$. Applying the submultiplicativity of the spectral norm to the product \eqref{eq:G} and using \eqref{eq:contractive} gives
\begin{equation}
\begin{aligned}
\norm{\bm g(u)}_2=\norm{\bm G(u)\bm f}_2 &\le\prod_{l=1}^{L}\norm{\bm\Phi_l(u)}_2\prod_{l=2}^{L}\norm{\bm W_l}_2\,\norm{\bm f}_2 \\& \le\prod_{l=2}^{L}\kappa_l\le1,
\end{aligned}
\end{equation}
since $\norm{\bm f}_2=1$; the bound holds for every $u$ and every phase trajectory. If all $\bm W_l$ are unitary, every factor of $\bm G(u)$ is unitary, hence $\bm G(u)$ is unitary and $\norm{\bm g(u)}_2=\norm{\bm f}_2=1$ for all $u$, which is the equality case. For the second claim, let the control be periodic with period $T_m$. Then $\bm g(u)$ is $T_m$-periodic, and it is bounded by the first claim, hence square integrable over one period. By Lemma~\ref{lem:conv}, its Fourier coefficients are the aperture harmonics $\bm g^{[p]}=T_m^{-1}\int_0^{T_m}\bm g(u)e^{-\jmathc2\pi pf_mu}du$. Applying Parseval's identity to each of the $N$ entries $[\bm g(u)]_n$ and summing over $n$ yields
\begin{equation}
\begin{aligned}
    \sum_{p\in\Z}\norm{\bm g^{[p]}}_2^2= \sum_{n=1}^N\sum_{p\in\Z}\abs{[\bm g^{[p]}]_n}^2 &=\frac1{T_m}\int_0^{T_m}\sum_{n=1}^N\abs{[\bm g(u)]_n}^2du\\&=\frac1{T_m}\int_0^{T_m}\norm{\bm g(u)}_2^2du.
\end{aligned}
\end{equation}
Since the integrand is at most one for every $u$ by the first claim, the time average is at most one, which completes the proof.
\end{IEEEproof}
Lemma~\ref{lem:energy} shows that temporal modulation redistributes the aperture energy across spectral modes rather than adding power: $\sum_q\norm{\bmg\psi_q}_2^2+T_m^{-1}\int_0^{T_m}\norm{\bm e_{\st}(u)}_2^2du\le1$: the desired modes and the unwanted sidebands share the same energy budget. Moreover, a lossless stack has constant energy $\norm{\bm g(u)}_2=1$, whereas a multi-tone field has time-varying energy (Lemma~\ref{lem:peak}); the difference is carried by $\bm e_{\st}(u)$ or absorbed by loss.

\section{Space--Frequency Synthesis and Doppler Compensation}\label{sec:ideal}
\subsection{Minimum-Energy Multiuser Synthesis}
Within one block, collect the steering vectors in $\bm A=[\bm a_1,\ldots,\bm a_K]\in\C^{N\times K}$ and let $\bm R_A\triangleq\bm A^H\bm A$. To cancel the Doppler of terminal $k$, one aperture mode $\bmg\psi_k$ is assigned the temporal factor $e^{-\jmathc2\pi f_{D,k}u}$:
\begin{equation}
\bm g(u)=\bm\Psi\bm d_D(u),\ \
\bm d_D(u)=\big[e^{-\jmathc2\pi f_{D,1}u},\ldots,e^{-\jmathc2\pi f_{D,K}u}\big]^T\!,
\label{eq:dD}
\end{equation}
with $\bm\Psi=[\bmg\psi_1,\ldots,\bmg\psi_K]$. Terminal $k$ must observe only mode $k$, i.e., $\bm a_k^H\bmg\psi_j=\beta_k\delta_{k,j}$, or
\begin{equation}
\bm A^H\bm\Psi=\bm D_\beta\triangleq\diag(\beta_1,\ldots,\beta_K),
\label{eq:decoupling}
\end{equation}
where $\beta_k\ne0$ prescribes the aperture gain toward terminal $k$. Among all decoupling matrices we select the one of minimum modal energy,
\begin{equation}
\mathcal P_1:\quad\min_{\bm\Psi\in\C^{N\times K}}\ \norm{\bm\Psi}_F^2\quad\mathrm{s.t.}\ \bm A^H\bm\Psi=\bm D_\beta,
\label{eq:P1}
\end{equation}
where $\norm{\bm\Psi}_F^2$ is also the long-time average energy of \eqref{eq:dD} for distinct mode frequencies, because the temporal exponentials are orthogonal in time average.

\begin{proposition}[Solution set, minimum-energy solution, and geometry bounds]\label{prop:solution}
Let $N\ge K$ and $\rank(\bm A)=K$. Then \eqref{eq:P1} is feasible, every feasible matrix has the form
\begin{equation}
\bm\Psi=\bm\Psi^\star+\bm P_A^\perp\bm Z,\qquad
\bm\Psi^\star=\bm A\bm R_A^{-1}\bm D_\beta,
\label{eq:Wgeneral}
\end{equation}
with $\bm Z\in\C^{N\times K}$ arbitrary and $\bm P_A^\perp=\bm I_N-\bm A\bm R_A^{-1}\bm A^H$ the orthogonal projector onto $\nulls(\bm A^H)$, and $\bm\Psi^\star$ is the unique minimum-Frobenius-norm solution, with $q_k\triangleq[\bm R_A^{-1}]_{k,k}$,
\begin{equation}
\norm{\bm\Psi^\star}_F^2=\tr\!\big(\bm D_\beta^H\bm R_A^{-1}\bm D_\beta\big)=\sum_{k=1}^Kq_k\abs{\beta_k}^2,
\label{eq:Wenergy}
\end{equation}
\begin{equation}
\frac{\norm{\bm D_\beta}_F^2}{\lambda_{\max}(\bm R_A)}\le\norm{\bm\Psi^\star}_F^2\le\frac{\norm{\bm D_\beta}_F^2}{\lambda_{\min}(\bm R_A)}.
\label{eq:energybounds}
\end{equation}
\end{proposition}
\begin{IEEEproof}
Since $\rank(\bm A)=K\le N$, the Gram matrix $\bm R_A=\bm A^H\bm A\in\C^{K\times K}$ is Hermitian positive definite and hence invertible, and $\bm A^H$ has full row rank, so its pseudoinverse is the right inverse $(\bm A^H)^\dagger=\bm A\bm R_A^{-1}$. \emph{Feasibility:} $\bm A^H\bm\Psi^\star=\bm A^H\bm A\bm R_A^{-1}\bm D_\beta=\bm D_\beta$, so $\bm\Psi^\star$ satisfies the constraint of \eqref{eq:P1}. \emph{Solution set:} the matrix $\bm P_A^\perp=\bm I_N-\bm A\bm R_A^{-1}\bm A^H$ is Hermitian and idempotent, and $\bm A^H\bm P_A^\perp=\bm A^H-\bm R_A\bm R_A^{-1}\bm A^H=\bm 0$, so it is the orthogonal projector onto $\nulls(\bm A^H)$. Hence $\bm A^H(\bm\Psi^\star+\bm P_A^\perp\bm Z)=\bm D_\beta$ for every $\bm Z$, i.e., every matrix of the form \eqref{eq:Wgeneral} is feasible. Conversely, if $\bm\Psi$ is feasible, then $\bm A^H(\bm\Psi-\bm\Psi^\star)=\bm 0$, so every column of $\bm\Psi-\bm\Psi^\star$ lies in $\nulls(\bm A^H)$ and $\bm\Psi-\bm\Psi^\star=\bm P_A^\perp(\bm\Psi-\bm\Psi^\star)$; choosing $\bm Z=\bm\Psi-\bm\Psi^\star$ shows that $\bm\Psi$ has the form \eqref{eq:Wgeneral}. \emph{Minimum norm and uniqueness:} the columns of $\bm\Psi^\star=\bm A(\bm R_A^{-1}\bm D_\beta)$ lie in $\range(\bm A)$, whereas those of $\bm P_A^\perp\bm Z$ lie in $\nulls(\bm A^H)=\range(\bm A)^\perp$. The cross term therefore vanishes, $(\bm\Psi^\star)^H\bm P_A^\perp\bm Z=\bm D_\beta^H\bm R_A^{-1}\bm A^H\bm P_A^\perp\bm Z=\bm 0$, and
\begin{equation*}
\norm{\bm\Psi}_F^2=\norm{\bm\Psi^\star}_F^2+\norm{\bm P_A^\perp\bm Z}_F^2\ge\norm{\bm\Psi^\star}_F^2,
\end{equation*}
with equality if and only if $\bm P_A^\perp\bm Z=\bm 0$, i.e., $\bm\Psi=\bm\Psi^\star$; thus $\bm\Psi^\star$ is the unique solution of \eqref{eq:P1}. \emph{Energy:} $(\bm\Psi^\star)^H\bm\Psi^\star=\bm D_\beta^H\bm R_A^{-1}\bm A^H\bm A\bm R_A^{-1}\bm D_\beta=\bm D_\beta^H\bm R_A^{-1}\bm D_\beta$, and taking the trace gives the first equality in \eqref{eq:Wenergy}; since $\bm D_\beta$ is diagonal, the $k$th diagonal entry of $\bm D_\beta^H\bm R_A^{-1}\bm D_\beta$ is $\abs{\beta_k}^2[\bm R_A^{-1}]_{k,k}=q_k\abs{\beta_k}^2$, which gives the second equality. \emph{Bounds:} the eigenvalues of $\bm R_A^{-1}$ are the reciprocals of those of $\bm R_A$, so $\lambda_{\max}^{-1}(\bm R_A)\bm I_K\preceq\bm R_A^{-1}\preceq\lambda_{\min}^{-1}(\bm R_A)\bm I_K$. Congruence by $\bm D_\beta$ preserves the Loewner order, which gives $\lambda_{\max}^{-1}(\bm R_A)\bm D_\beta^H\bm D_\beta\preceq\bm D_\beta^H\bm R_A^{-1}\bm D_\beta\preceq\lambda_{\min}^{-1}(\bm R_A)\bm D_\beta^H\bm D_\beta$, and the trace is monotone in this order. With $\tr(\bm D_\beta^H\bm D_\beta)=\norm{\bm D_\beta}_F^2$, this yields \eqref{eq:energybounds}.
\end{IEEEproof}

Equation~\eqref{eq:energybounds} gives terminal geometry a direct physical meaning: when two steering vectors become nearly collinear, $\lambda_{\min}(\bm R_A)\to0$ and the energy needed to separate their frequency translations grows without bound. \emph{Spatial separability and Doppler compensation are therefore fundamentally coupled.} Passivity (Lemma~\ref{lem:energy}) constrains the instantaneous rather than the average energy, which motivates the following definition.
\begin{lemma}[Long-run peak energy of the modal field]\label{lem:peak}
For any $\bm\Psi$ and rationally independent $f_{D,k}$, the trajectory $\bm d_D(u)$ is dense in the phase torus, so that
\begin{equation}
P_{\rm pk}(\bm\Psi)\triangleq\sup_u\norm{\bm\Psi\bm d_D(u)}_2^2=\max_{\bmg\phi\in[0,2\pi)^K}\norm{\bm\Psi e^{\jmathc\bmg\phi}}_2^2
\label{eq:peak}
\end{equation}
satisfies $\norm{\bm\Psi}_F^2\le P_{\rm pk}(\bm\Psi)\le\big(\sum_k\norm{\bmg\psi_k}_2\big)^2\le K\norm{\bm\Psi}_F^2$, with the long-time average equal to $\norm{\bm\Psi}_F^2$.
\end{lemma}
\begin{IEEEproof}
Density follows from Kronecker's theorem; the upper bounds are the triangle and Cauchy--Schwarz inequalities, and the lower bound and the average follow from the orthogonality of the exponentials.
\end{IEEEproof}
The phases of $\beta_k$ do not change $P_{\rm pk}$, since the trajectory sweeps every phase combination. 

\subsection{Perfect Blockwise Multiuser Compensation}
For arbitrary temporal frequencies $\bm\nu=[\nu_1,\ldots,\nu_K]^T$ let $\bm g_{\bm\nu}(u)=\bm\Psi\bm d_{\bm\nu}(u)$ with $\bm d_{\bm\nu}(u)=[e^{-\jmathc2\pi\nu_1u},\ldots,e^{-\jmathc2\pi\nu_Ku}]^T$. Substituting into \eqref{eq:blockrx} and absorbing fixed delay phases into the coefficients yields the \emph{space--frequency coupling identity}
\begin{equation}
\frac{y_k(t)-n_k(t)}{\alpha_kx(t-\tau_k)}=\sum_{j=1}^KH_{k,j}\,e^{\jmathc2\pi(f_{D,k}-\nu_j)t}e^{\jmathc2\pi\nu_j\tau_k},
\label{eq:coupling}
\end{equation}
where $\bm H\triangleq\bm A^H\bm\Psi$. The identity factorizes the two design dimensions: $H_{k,j}$ is the spatial leakage of temporal mode $j$ into terminal $k$, and $f_{D,k}-\nu_j$ is the residual temporal rotation of that mode.

\begin{theorem}[Perfect blockwise multiuser Doppler compensation]\label{thm:perfect}
Let $\beta_k\ne0$ for all $k$, $\rank(\bm A)=K$, and suppose the ST-SIM realizes $\bm g^\star(u)=\bm\Psi^\star\bm d_D(u)$ with $\bm\Psi^\star$ from \eqref{eq:Wgeneral}. Then every terminal has zero residual Doppler and zero inter-mode leakage:
\begin{equation}
y_k(t)=\alpha_k\beta_ke^{\jmathc2\pi f_{D,k}\tau_k}x(t-\tau_k)+n_k(t),\quad k=1,\ldots,K.
\label{eq:perfectrx}
\end{equation}
Conversely, exact decoupling $\bm A^H\bm\Psi=\bm D_\beta$ with a nonsingular diagonal $\bm D_\beta$ requires $\rank(\bm A)=K$; if the $f_{D,k}$ are distinct, the rank condition is also necessary for \eqref{eq:perfectrx} to hold for all $k$.
\end{theorem}
\begin{IEEEproof}
By construction $H_{k,j}=\beta_k\delta_{k,j}$; choosing $\nu_j=f_{D,j}$ in \eqref{eq:coupling} leaves only the $j=k$ term with unit temporal factor, which is \eqref{eq:perfectrx}. Conversely, $K=\rank(\bm D_\beta)\le\rank(\bm A^H)\le K$ forces $\rank(\bm A)=K$. If the $f_{D,k}$ are distinct, the exponentials in \eqref{eq:coupling} are linearly independent functions of $t$, so \eqref{eq:perfectrx} forces $H_{k,j}=\beta_k\delta_{k,j}$ for every $k$ and $j$, i.e., exact decoupling, and the previous argument applies.
\end{IEEEproof}

Theorem~\ref{thm:perfect} identifies the sufficient aperture field and the rank condition for its existence; whether a finite-depth phase-only stack synthesizes it is addressed in Section~\ref{sec:practical}.
\begin{remark}[Grouped compensation]\label{rem:grouped}
The rank condition concerns $K$ mutually decoupled translations. If the terminals are partitioned into $G\le K$ groups $\mathcal K_q$ sharing one tone $\nu_q$ and one multicast mode $\bmg\psi_q$ each, decoupling requires $\bm A^H\bm\Psi=\bm D_\beta\bm E$ with the group-indicator matrix $\bm E\in\{0,1\}^{K\times G}$, which is consistent whenever $\rank(\bm A)=K$ and possibly otherwise (e.g., for collinear terminals in one group), with minimum-norm solution $(\bm A^H)^\dagger\bm D_\beta\bm E$, i.e., $\bmg\psi_q=\sum_{k\in\mathcal K_q}\bmg\psi_k^\star$ under full rank. By Proposition~\ref{prop:common} per group, each terminal keeps at most half the Doppler span of its group; $G=1$ is the common translation and $G=K$ is Theorem~\ref{thm:perfect} (Section~\ref{sec:res_group}). Terminals should thus be grouped contiguously in Doppler, i.e., by along-track position: a group of Doppler span $\Delta_q$ keeps a residual $\epsilon=\Delta_q/(2\Delta f)$ with interference-limited SIR $\approx3/(\pi^2\epsilon^2)$ (Proposition~\ref{prop:ici}), whereas every additional group costs modal energy (Proposition~\ref{prop:solution}), a trade-off evaluated in Section~\ref{sec:res_group}.
\end{remark}

\subsection{Irreducible Residual of a Common Frequency Translation}
A conventional payload with one RF chain and one mixer, a static SIM, or a time-modulated surface that applies one translation to all beams \cite{He2015,Zhu2024}, imposes a single translation $\nu_c$ on the common waveform, with worst-terminal residual $\epsilon_\infty(\nu_c)=\max_k\abs{f_{D,k}-\nu_c}$.
\begin{proposition}[Optimal common precompensation]\label{prop:common}
With $f_D^{\min}=\min_kf_{D,k}$ and $f_D^{\max}=\max_kf_{D,k}$, the minimax translation and residual are
\begin{equation}
\nu_c^\star=\tfrac12\big(f_D^{\max}+f_D^{\min}\big),\qquad
\epsilon_\infty^\star=\tfrac12\big(f_D^{\max}-f_D^{\min}\big),
\label{eq:spanbound}
\end{equation}
and the least-squares translation is $\nu_c^{\rm LS}=K^{-1}\sum_kf_{D,k}$. Both minimize a residual rather than the SINR; the SINR-optimal translation $\nu_c^{\rm SINR}=\arg\max_{\nu_c}\min_k\gamma_k(\nu_c)$, a one-dimensional search over $[f_D^{\min},f_D^{\max}]$, is evaluated in Section~\ref{sec:res_link}.
\end{proposition}
\begin{IEEEproof}
Since the maximum of finitely many absolute values is attained at an extreme point, $\epsilon_\infty(\nu_c)=\max\{f_D^{\max}-\nu_c,\ \nu_c-f_D^{\min}\}$, a convex piecewise-linear function with slopes $-1$ and $+1$ whose minimum is attained where the two arguments coincide, i.e., at the midpoint $\nu_c^\star$, with value $\epsilon_\infty^\star$; any $\nu_c=\nu_c^\star+\delta$ incurs $\epsilon_\infty^\star+\abs\delta$, so the minimizer is unique. The least-squares criterion $\sum_k(f_{D,k}-\nu_c)^2$ is strictly convex, and its derivative vanishes at $\nu_c^{\rm LS}$.
\end{IEEEproof}
Whenever the Doppler span is nonzero, no common translation makes every residual vanish; Theorem~\ref{thm:perfect} removes this floor by pairing a different translation with each spatial mode, and the comparison is architectural rather than algorithmic. Since the payload contains a frequency synthesizer, the numerical results apply $\nu_c^\star$ to the feed and let the ST-SIM synthesize only the differential tones $f_{D,k}-\nu_c^\star$; all results hold for arbitrary tone frequencies, and this hybrid mode reduces the largest synthesized tone from $\max_k\abs{f_{D,k}}$ to $\epsilon_\infty^\star$, by a factor of $1.05$ in the baseline geometry and more for clustered terminals.

\subsection{OFDM Consequences and Residual Errors}\label{sec:ofdm}
Consider $M$ subcarriers with spacing $\Delta f$ and useful duration $T=1/\Delta f$, and let terminal $k$ have residual frequency error $\delta f_k=f_{D,k}-\nu_k$ with normalized value $\epsilon_k=\delta f_k/\Delta f$. After cyclic-prefix removal, the DFT-domain relation is $Y_k[p]=h_k\sum_qC_{p-q}(\epsilon_k)X_k[q]+N_k[p]$ with
\begin{equation}
C_r(\epsilon)=\frac1M\sum_{m=0}^{M-1}e^{\jmathc2\pi(\epsilon-r)\frac mM}
=\frac{e^{\jmathc\pi\frac{(M-1)(\epsilon-r)}{M}}}{M}\frac{\sin[\pi(\epsilon-r)]}{\sin[\frac{\pi(\epsilon-r)}M]}.
\label{eq:Crclosed}
\end{equation}
The CFO matrix is unitary, so $\sum_r\abs{C_r(\epsilon)}^2=1$, and under Theorem~\ref{thm:perfect} $\epsilon_k=0$, $C_0(0)=1$, and the carrier-Doppler-induced ICI vanishes for every terminal. For imperfect compensation, the fraction of desired-mode power leaving the intended subcarrier is
\begin{equation}
\chi_M(\epsilon)\triangleq1-\abs{C_0(\epsilon)}^2=1-\left|\frac{\sin(\pi\epsilon)}{M\sin(\pi\epsilon/M)}\right|^2.
\label{eq:ici}
\end{equation}
\begin{proposition}[Local ICI law]\label{prop:ici}
As $\epsilon\to0$, $\chi_M(\epsilon)=\frac{\pi^2}{3}\big(1-\frac1{M^2}\big)\epsilon^2+\mathcal O(\epsilon^4)$.
\end{proposition}
\begin{IEEEproof}
Expand $\sin x=x-x^3/6+\mathcal O(x^5)$ in \eqref{eq:ici}: $\sin(\pi\epsilon)/[M\sin(\pi\epsilon/M)]=1-\frac{\pi^2}{6}(1-M^{-2})\epsilon^2+\mathcal O(\epsilon^4)$; square and subtract from one.
\end{IEEEproof}
Since the common-compensation residual scales linearly with $f_c$ and the ICI fraction quadratically with the residual normalized by $\Delta f$, differential Doppler becomes critical as $f_c/\Delta f$ grows (Section~\ref{sec:res_link}). Two residuals remain after carrier-Doppler compensation. First, the envelope in \eqref{eq:exactrx} is time-scaled by $1+\alpha_k$, $\alpha_k\approx f_{D,k}/f_c$, which leaves the band-averaged ICI fraction $\bar\chi_k\approx(\pi^2/36)(Mf_{D,k}/f_c)^2$, i.e., $-47$ and $-38$~dB at $160$ and $450$~kHz; this sampling-clock-offset effect is common to all carrier-Doppler schemes and is removed by the timing tracking of an OFDM receiver. Second, ephemeris and terminal-position errors leave a residual tone. With the first-order model $f_{D,k}=-(f_c/c)\bm u_k^T\bm v_k^{\rm rel}$, suppose the controller uses $\widehat{\bm\rho}_k=\bm\rho_k+\delta\bm\rho_k$ and $\widehat{\bm v}_k^{\rm rel}=\bm v_k^{\rm rel}+\delta\bm v_k^{\rm rel}$ to predict $\widehat f_{D,k}$, with $\bm\rho_k=d_k\bm u_k$ the range vector, and let $\bm P_k^\perp=\bm I_3-\bm u_k\bm u_k^T$.
\begin{proposition}[Residual Doppler under state errors]\label{prop:ephemeris}
If $\norm{\delta\bm\rho_k}_2<d_k$, then
\begin{equation}
\abs{\widehat f_{D,k}-f_{D,k}}\le\frac{f_c}{c}\left[\norm{\delta\bm v_k^{\rm rel}}_2+\frac{2\norm{\bm v_k^{\rm rel}}_2\norm{\delta\bm\rho_k}_2}{d_k-\norm{\delta\bm\rho_k}_2}\right],
\label{eq:nonasymfd}
\end{equation}
and, to first order, $\delta f_{D,k}\triangleq\widehat f_{D,k}-f_{D,k}=-\frac{f_c}{c}\big[\bm u_k^T\delta\bm v_k^{\rm rel}+d_k^{-1}(\bm v_k^{\rm rel})^T\bm P_k^\perp\delta\bm\rho_k\big]$.
\end{proposition}
\begin{IEEEproof}
With $\widehat{\bm u}_k=(\bm\rho_k+\delta\bm\rho_k)/\norm{\bm\rho_k+\delta\bm\rho_k}_2$, $\widehat f_{D,k}-f_{D,k}=-(f_c/c)[(\widehat{\bm u}_k-\bm u_k)^T\bm v_k^{\rm rel}+\widehat{\bm u}_k^T\delta\bm v_k^{\rm rel}]$; the reverse triangle inequality gives $\norm{\widehat{\bm u}_k-\bm u_k}_2\le2\norm{\delta\bm\rho_k}_2/(d_k-\norm{\delta\bm\rho_k}_2)$, and Cauchy--Schwarz yields \eqref{eq:nonasymfd}. The first-order form follows from $d\bm u_k=d_k^{-1}\bm P_k^\perp d\bm\rho_k$ and the product rule.
\end{IEEEproof}
If the ST-SIM applies $\nu_k=\widehat f_{D,k}$, the residual CFO is $-\delta f_{D,k}$ and, by Proposition~\ref{prop:ici}, the ICI fraction is $\frac{\pi^2}{3}(1-M^{-2})(\delta f_{D,k}/\Delta f)^2$, which exposes the critical factor $f_c/\Delta f$ of high-frequency LEO OFDM links. Only the \emph{transverse} component of a position error matters, scaled by $\norm{\bm v^{\rm rel}}_2/d_k$.

\section{Phase-Only ST-SIM Synthesis}\label{sec:practical}
The ideal field of Theorem~\ref{thm:perfect} need not belong to the set of fields realizable by a finite-depth phase-only stack. This section formulates the physical synthesis problem and derives an efficient first-order method with explicit guarantees.

\subsection{Trajectory Synthesis Problem}
Sample a predictable control horizon at transmit times $u_r=u_0+rT_u$, $r=0,\ldots,R-1$, where $f_u=1/T_u$ is the control update rate, and let $\bm g_r^\star$ be the desired field generated by \eqref{eq:dD}; since $\bm\Psi^\star$ has full column rank and $\bm d_D(u)\ne\bm 0$, $\norm{\bm g_r^\star}_2>0$ for every $r$. Let $\bm\theta[r]\in\R^{LN}$ collect all layer phases at slot $r$ and $\Theta=[\bm\theta[0],\ldots,\bm\theta[R-1]]$. A passive stack cannot reach an arbitrary absolute amplitude (Lemma~\ref{lem:energy}), so a common complex scaling $\alpha$, which absorbs the transmission gain of the stack, is optimized jointly with the phases. The finite-resolution design problem is
\begin{equation}
\begin{aligned}
\mathcal P_2:\ \min_{\Theta,\alpha}\ &\mathcal J(\Theta,\alpha)=\frac1R\sum_{r=0}^{R-1}\omega_r\frac{\norm{\alpha\bm G(\bm\theta[r])\bm f-\bm g_r^\star}_2^2}{\norm{\bm g_r^\star}_2^2}\\
&\qquad\qquad+\lambda_s\mathcal S(\Theta)\\
\mathrm{s.t.}\ &\theta_{l,n}[r]\in\Qphi,
\end{aligned}
\label{eq:P2}
\end{equation}
where $\Qphi=\{2\pi q/2^{B_\phi}\}_{q=0}^{2^{B_\phi}-1}$ is the $B_\phi$-bit phase alphabet, $\omega_r>0$ are trajectory weights, and
\begin{equation}
\mathcal S(\Theta)=\frac1{R-1}\sum_{r=1}^{R-1}\sum_{i=1}^{LN}\Big[1-\cos\big(\theta_i[r]-\theta_i[r-1]\big)\Big]
\label{eq:switchreg}
\end{equation}
penalizes switching between consecutive slots. The penalty \eqref{eq:switchreg} is $2\pi$-periodic and smooth, unlike a penalty on wrapped phase differences, so it does not affect the differentiability assumed below; its gradient is $\partial\mathcal S/\partial\theta_i[r]=\frac{1}{R-1}[\sin(\theta_i[r]-\theta_i[r-1])-\sin(\theta_i[r+1]-\theta_i[r])]$. A hard slew limit, if required by the drivers, can be imposed afterwards by projecting each increment onto $[-\Delta_{\max},\Delta_{\max}]$. Problem~\eqref{eq:P2} is nonconvex because of the cascaded phase products and the discrete alphabet. We study its continuous relaxation, then quantify the effect of quantization.

\subsection{Optimal Scaling and Reusable Multilayer Gradient}
For one slot (index $r$ suppressed), define the field incident on layer $l$ and the transfer from the output of layer $l$ to the aperture, with $\bm u_1=\bm f$,
\begin{equation}
\bm u_l=\bm W_l\bm\Phi_{l-1}\bm u_{l-1},\ \
\bm V_l=\begin{cases}\bm I_N,&l=L,\\ \bm\Phi_L\bm W_L\cdots\bm\Phi_{l+1}\bm W_{l+1},&l<L,\end{cases}
\label{eq:fwdbwd}
\end{equation}
so that $\bm g=\bm V_l\bm\Phi_l\bm u_l$ for every $l$.
\begin{lemma}[Optimal scaling and phase derivatives]\label{lem:grad}
For fixed $\Theta$, with $\bm g_r=\bm G(\bm\theta[r])\bm f$ and $w_r=\omega_r/\norm{\bm g_r^\star}_2^2$, the minimizer of $\mathcal J$ over $\alpha$ is
\begin{equation}
\alpha^\star(\Theta)=\frac{\sum_rw_r\bm g_r^H\bm g_r^\star}{\sum_rw_r\norm{\bm g_r}_2^2},
\label{eq:alphastar}
\end{equation}
and $\tilde{\mathcal J}(\Theta)\triangleq\mathcal J(\Theta,\alpha^\star(\Theta))$ has gradient $\nabla_\Theta\tilde{\mathcal J}=\nabla_\Theta\mathcal J|_{\alpha=\alpha^\star}$. For meta-atom $n$ of layer $l$, $\partial\bm g/\partial\theta_{l,n}=\jmathc e^{\jmathc\theta_{l,n}}[\bm u_l]_n\bm V_l\bm e_n$, and for $F(\bm\theta)=\norm{\alpha\bm g(\bm\theta)-\bm g^\star}_2^2$,
\begin{equation}
\frac{\partial F}{\partial\theta_{l,n}}=2\Re\!\left\{\jmathc e^{\jmathc\theta_{l,n}}[\bm u_l]_n\,[\bm V_l^H\bm\varepsilon]_n^*\right\},\quad
\bm\varepsilon\triangleq\alpha^*(\alpha\bm g-\bm g^\star).
\label{eq:objectivegradient}
\end{equation}
\end{lemma}
\begin{IEEEproof}
$\mathcal J$ is a convex quadratic in $\alpha$ whose stationary point is \eqref{eq:alphastar}; the gradient identity is the envelope theorem, since $\partial\mathcal J/\partial\alpha^*=0$ at $\alpha^\star$. $\partial\bm\Phi_l/\partial\theta_{l,n}=\jmathc e^{\jmathc\theta_{l,n}}\bm e_n\bm e_n^T$ gives the derivative of $\bm g$, and $dF=2\Re\{(\alpha\bm g-\bm g^\star)^H\alpha\,d\bm g\}$ gives \eqref{eq:objectivegradient}.
\end{IEEEproof}
All $LN$ derivatives of a slot follow from one forward sweep for $\{\bm u_l\}$ and one backward sweep for $\{\bm V_l^H\bm\varepsilon\}$, i.e., $2(L-1)$ matrix--vector products of size $N$, without differentiating the full cascade for every meta-atom.

\subsection{Algorithm and Convergence}
Algorithm~\ref{alg:synthesis} summarizes the proposed synthesis. Step~1 (lines 2--8) performs wrapped gradient descent on the continuous relaxation, where wrapping only changes the coordinate representation of a point on the phase torus; the trial step is the Barzilai--Borwein estimate of the inverse local curvature \cite{Barzilai1988}, safeguarded to $[\mu_{\min},\mu_{\max}]$ and then reduced by Armijo backtracking until sufficient decrease is achieved. Step~2 (line 9) projects onto $\Qphi$, and Step~3 (lines 10--12) performs quantization-aware coordinate refinement, in which each meta-atom is assigned the best of its $2^{B_\phi}$ states given all others, using $\bm V_l$ and $\bm u_l$ so that each trial costs $\mathcal O(N)$. With $\alpha^\star$ fixed, the coordinate cost of $\theta_{l,n}[r]$ collects all terms of $\mathcal J$ that depend on it,
\begin{equation}
\begin{aligned}
\mathcal J_{l,n,r}(\phi)&=\tfrac{w_r}{R}\big\|\alpha^\star\bm V_l(\bm\Phi_l\bm u_l+\Delta_\phi[\bm u_l]_n\bm e_n)-\bm g_r^\star\big\|_2^2\\
&\quad+\tfrac{\lambda_s}{R-1}\big[2-\cos(\phi-\widehat\theta^-)-\cos(\widehat\theta^+-\phi)\big],
\end{aligned}
\label{eq:coordcost}
\end{equation}
where $\Delta_\phi=e^{\jmathc\phi}-e^{\jmathc\widehat\theta_{l,n}[r]}$, $\widehat\theta^\pm=\widehat\theta_{l,n}[r\pm1]$, and the switching terms are omitted at $r=0$ and $r=R-1$. For $\lambda_s=0$ the slots decouple and are updated in parallel; for $\lambda_s>0$ the even and the odd slots are updated in alternate half-sweeps, so that every update sees fixed neighbors.

\begin{algorithm}[t]
\caption{Phase-only ST-SIM trajectory synthesis}
\label{alg:synthesis}
\begin{algorithmic}[1]
\REQUIRE $\{\bm g_r^\star\}_{r=0}^{R-1}$, $\{\bm W_l\}$, $\bm f$, $B_\phi$, $0<c<1$, $0<\rho<1$, $I_{\max}$, $n_{\rm ref}$
\STATE Initialize $\Theta^{(0)}$ uniformly at random on $[0,2\pi)^{LN\times R}$
\FOR{$i=0,1,\ldots,I_{\max}-1$}
\STATE Forward sweeps: $\{\bm u_l[r]\}$, $\bm g_r$, and $\alpha^\star$ by \eqref{eq:alphastar}
\STATE Backward sweeps: $\bm\Gamma_i\triangleq\nabla\tilde{\mathcal J}(\Theta^{(i)})$ by \eqref{eq:objectivegradient}
\STATE Trial step $\mu\gets\langle\bm S_i,\bm S_i\rangle/\langle\bm S_i,\bm\Gamma_i-\bm\Gamma_{i-1}\rangle$, $\bm S_i=\Wwrap(\Theta^{(i)}-\Theta^{(i-1)})$, clipped to $[\mu_{\min},\mu_{\max}]$
\STATE \textbf{while} $\tilde{\mathcal J}(\Wwrap(\Theta^{(i)}-\mu\bm\Gamma_i))>\tilde{\mathcal J}(\Theta^{(i)})-c\mu\norm{\bm\Gamma_i}_F^2$ \textbf{do} $\mu\gets\rho\mu$
\STATE $\mu_i\gets\mu$;\ $\Theta^{(i+1)}\gets\Wwrap(\Theta^{(i)}-\mu_i\bm\Gamma_i)$
\ENDFOR
\STATE Quantize: $\widehat\Theta\gets\arg\min_{\Theta\in\Qphi^{LN\times R}}\norm{\Wwrap(\Theta-\Theta^{(I_{\max})})}_F$
\FOR{$n_{\rm ref}$ rounds, every slot $r$, layer $l$, meta-atom $n$}
\STATE $\widehat\theta_{l,n}[r]\gets\arg\min_{\phi\in\Qphi}\mathcal J_{l,n,r}(\phi)$ in \eqref{eq:coordcost} \hfill (even $r$, then odd $r$)
\ENDFOR
\ENSURE $\widehat\Theta$, $\alpha^\star$
\end{algorithmic}
\end{algorithm}

\begin{proposition}[Stationarity of Armijo descent]\label{prop:armijo}
Let $\tilde{\mathcal J}=\mathcal J(\cdot,\alpha^\star(\cdot))$ with $\lambda_s\ge0$ and the smooth penalty \eqref{eq:switchreg}, which is continuously differentiable with Lipschitz gradient on the compact phase torus (the data term is a smooth function of unit-modulus entries and $\sum_r\norm{\bm g_r}_2^2>0$ since $\bm f\ne\bm 0$ and the $\bm W_l$ are nonsingular), and let $\{\Theta^{(i)}\}$ be generated by Step~1 of Algorithm~\ref{alg:synthesis} with trial steps in $[\mu_{\min},\mu_{\max}]$, $0<\mu_{\min}\le\mu_{\max}<\infty$. Then $\tilde{\mathcal J}(\Theta^{(i+1)})\le\tilde{\mathcal J}(\Theta^{(i)})$, $\sum_i\mu_i\norm{\nabla\tilde{\mathcal J}(\Theta^{(i)})}_F^2<\infty$, and every accumulation point is first-order stationary; in particular $\lim_i\norm{\nabla\tilde{\mathcal J}(\Theta^{(i)})}_F=0$ whenever the accepted step sizes are bounded away from zero along nonstationary subsequences. Every update of Step~3 is nonincreasing in $\mathcal J(\cdot,\alpha^\star)$ for any $\lambda_s\ge0$.
\end{proposition}
\begin{IEEEproof}
Lipschitz continuity implies the descent lemma on any local chart of the torus, so backtracking from any trial step in $[\mu_{\min},\mu_{\max}]$ accepts, after finitely many reductions, a step bounded below by $\min\{\mu_{\min},2\rho(1-c)/L_{\nabla}\}$, with $L_\nabla$ the Lipschitz constant, and $\tilde{\mathcal J}(\Theta^{(i)})-\tilde{\mathcal J}(\Theta^{(i+1)})\ge c\mu_i\norm{\nabla\tilde{\mathcal J}(\Theta^{(i)})}_F^2$. The objective is bounded below, so the series converges; compactness yields accumulation points, and a nonstationary one would contradict the finiteness of the series together with the Armijo lower step bound in its neighborhood \cite{Bertsekas1999}. In Step~3, \eqref{eq:coordcost} contains every term of $\mathcal J(\cdot,\alpha^\star)$ that depends on the updated phase, with fixed neighbors; each update is thus exact coordinate descent on a finite set, and recomputing $\alpha^\star$ can only decrease $\mathcal J$ further.
\end{IEEEproof}
The trial step is well defined for every curvature estimate: when $\langle\bm S_i,\bm\Gamma_i-\bm\Gamma_{i-1}\rangle\le0$, the Barzilai--Borwein quotient is replaced by $\mu_{\max}$, which the backtracking then reduces until sufficient decrease holds, so the guarantee of the proposition is unaffected. In the baseline geometry of Section~\ref{sec:results}, $\norm{\nabla\tilde{\mathcal J}}_F$ after $300$ iterations is $2.5\times10^{-5}$, $7.5\times10^{-6}$, $2.6\times10^{-6}$, and $9.2\times10^{-6}$ for $L=2$, $3$, $4$, and $6$, and the single layer meets the tolerance $10^{-6}$ after $142$ iterations: its stall at $-5.9$~dB is a stationary point. The proposition guarantees monotone decrease and convergence to first-order stationary points, without global optimality, for the continuous relaxation of $\mathcal P_2$; the discrete alphabet is handled by Steps~2--3. Each iteration and backtracking trial cost $\mathcal O(RLN^2)$.

\subsection{Finite-Resolution Phase Quantization}
Let $\widehat{\bm\Phi}_l$ denote the nearest-neighbor $B_\phi$-bit quantization of $\bm\Phi_l$ and $\widehat{\bm G}$ the corresponding stack matrix, and define $\delta_{B_\phi}\triangleq2\sin(\pi/2^{B_\phi+1})$, so that $\norm{\widehat{\bm\Phi}_l-\bm\Phi_l}_2\le\delta_{B_\phi}$.
\begin{proposition}[Multilayer quantization scaling law]\label{thm:quant}
With $\kappa_l=\norm{\bm W_l}_2$,
\begin{equation}
\norm{(\widehat{\bm G}-\bm G)\bm f}_2\le L\,\delta_{B_\phi}\Big(\prod_{l=2}^L\kappa_l\Big)\norm{\bm f}_2\le L\,\delta_{B_\phi}\norm{\bm f}_2,
\label{eq:gquantbound}
\end{equation}
where the last inequality holds for a passive stack.
\end{proposition}
\begin{IEEEproof}
Write $\widehat{\bm G}-\bm G$ as a telescoping sum of $L$ terms in which one phase matrix at a time is replaced by its quantized version. Each term contains exactly one factor $\widehat{\bm\Phi}_l-\bm\Phi_l$, $L-1$ unit-norm phase matrices, and all propagation matrices once, hence has norm at most $\delta_{B_\phi}\prod_l\kappa_l$.
\end{IEEEproof}
The bound ignores cancellation among layers and is loose for coarse quantization, where it exceeds the norm of the field itself ($L\delta_{B_\phi}>1$ for $B_\phi\le4$ with $L=4$); it is to be read as a scaling law: for fixed depth the field error decays as $\mathcal O(2^{-B_\phi})$, i.e., the synthesis error $\mathcal J$, which is a squared norm, as $\mathcal O(2^{-2B_\phi})$ or $6$~dB per bit, and for fixed resolution its worst-case accumulation is at most linear in $L$ (Fig.~\ref{fig:synthesis_error}).

\subsection{Control Rate, Switching Harmonics, and Control Bandwidth}
If the desired field contains temporal frequencies in $[-\nu_{\max},\nu_{\max}]$ and the control is interpolated ideally, sampling theory requires $f_u>2\nu_{\max}$. Practical digital-coding metasurfaces, however, hold each state constant over $T_u$, which generates switching harmonics.
\begin{proposition}[Control-rate law for held control]\label{prop:zoh}
Let the aperture field be held constant over each control slot, $\bm g_h(u)=\bm g^\star(u_r)$ for $u\in[u_r,u_r+T_u)$, and let the ideal field toward terminal $k$ be the tone $\beta_ke^{-\jmathc2\pi f_{D,k}u}$. Then the held field toward terminal $k$ contains the tone at $-f_{D,k}$ with amplitude $\beta_k\sinc(f_{D,k}/f_u)$ and images at $-f_{D,k}-pf_u$, $p\in\Z\setminus\{0\}$, with amplitudes $\beta_k\sinc\big((f_{D,k}+pf_u)/f_u\big)$. The energy fraction diverted into the images is
\begin{equation}
1-\sinc^2\!\big(f_{D,k}/f_u\big)=\frac{\pi^2}{3}\Big(\frac{f_{D,k}}{f_u}\Big)^2+\mathcal O\big((f_{D,k}/f_u)^4\big),
\label{eq:zohlaw}
\end{equation}
so that, when the images fall inside the signal band, the self-interference ratio of terminal $k$ is capped at
$\mathrm{SIR}_k^{\rm ZOH}=\sinc^2(f_{D,k}/f_u)/[1-\sinc^2(f_{D,k}/f_u)]\approx3f_u^2/(\pi^2f_{D,k}^2)$.
\end{proposition}
\begin{IEEEproof}
The held signal is the sampled tone convolved with a rectangular pulse of length $T_u$, whose spectrum is $T_ue^{-\jmathc\pi fT_u}\sinc(fT_u)$, evaluated at the aliases $-f_{D,k}-pf_u$ of the sampled tone; the held tone has unit modulus, so by Parseval the image energy equals one minus the fundamental energy, and the expansion of $\sinc^2$ gives \eqref{eq:zohlaw}.
\end{IEEEproof}
Equation~\eqref{eq:zohlaw} has the same rectangular-window origin as Proposition~\ref{prop:ici} and yields an explicit hardware requirement: an in-band SIR of 30~dB for a synthesized tone of $160$~kHz requires $f_u\gtrsim9$~MHz, and when $f_u$ exceeds the signal bandwidth the images fall out of band, where they are an emission rather than an interference problem (Fig.~\ref{fig:control}).

The binding hardware constraint is not the MHz-rate switching of single elements, which is established \cite{Zhang2018,Zhang2021,Sun2025}, but the \emph{aggregate} control bandwidth: streaming every state requires $LNB_\phi f_u$ bit/s, i.e., $9.2$~Gbit/s for $LN=400$, $B_\phi=3$, and $f_u=7.68$~MHz. If the tones are rounded to a grid $\delta_f$, the target is periodic with period $1/\delta_f$, and each meta-atom can replay a local table of $f_u/\delta_f$ states, provided that one full period is synthesized offline, since a nonconvex synthesis does not return a periodic trajectory by itself; for $\delta_f=100$~Hz and $f_u=7.68$~MHz, a period comprises $76\,800$ slots ($3.1\times10^7$ phases for $LN=400$), which is tractable because \eqref{eq:P2} couples the slots only through the switching penalty and the gain $\alpha$: synthesized in parallel blocks of $R=274$ slots with fixed $\alpha$ and overlapping boundaries ($4$~s per block on one core), the period costs about $20$~core-minutes per pass geometry. The tables are refreshed when a rounded tone has drifted by $\delta_f/2$, which reduces the streaming rate by $\delta_f^2/(2\max_k\abs{\dot f_{D,k}-\dot\nu_c^\star})$, i.e., with the drift of at most $1.5$~kHz/s of Section~\ref{sec:results}, by $3.3\times$ for $\delta_f=100$~Hz and $30\times$ for $\delta_f=300$~Hz (residual below $0.005\Delta f$).

\section{Numerical Results}\label{sec:results}
In this section, we extensively investigate the performance of the proposed ST-SIM architecture.
\subsection{Setup, Benchmarks, and Metrics}
Table~\ref{tab:sim} lists the baseline parameters. At $t=0$, the four baseline terminals at ground offsets $(-180,60)$, $(-40,-130)$, $(90,140)$, and $(200,-30)$~km (along-, cross-track) see Doppler shifts of $-143.7$, $-32.7$, $72.7$, and $158.6$~kHz. Geometry follows Section~\ref{sec:model} with static terminals in the free-space LoS channel of \eqref{eq:exactrx}, with steering vectors and Doppler shifts evaluated at the actual directions and slant ranges; the path gain $\ell_k$ is absorbed into $\mathrm{SNR}_0$ and taken equal for all terminals, since the slant range grows from $600$~km at the beam center to $677$~km at $300$~km from it, i.e., by at most $1.1$~dB of free-space loss; $\bm W_l$ and $\bm f$ follow the Rayleigh--Sommerfeld model of \cite{AnJSAC2023,Lin2018} with $\norm{\bm W_l}_2=1$. The payload broadcasts one common waveform, and the control trajectory is synthesized over the OFDM symbol including its cyclic prefix ($72$ samples) and held over $T_u=4$ samples, so switching harmonics are included; the ST-SIM synthesizes the differential tones $f_{D,k}-\nu_c^\star$. Algorithm~\ref{alg:synthesis} uses $[\mu_{\min},\mu_{\max}]=[10^{-2},10^{5}]$, $\lambda_s=0$, a gradient tolerance of $10^{-6}$, $n_{\rm ref}=2$, and $I_{\max}= 300$. Figures show medians over $50$ random drops unless stated.

\begin{table}[t]
\centering
\caption{Baseline Simulation Parameters}
\label{tab:sim}
\scriptsize
\setlength{\tabcolsep}{3pt}
\renewcommand{\arraystretch}{0.95}
\begin{tabular}{l c l c}
\toprule
Parameter & Value & Parameter & Value\\
\midrule
Carrier $f_c$ & 20 GHz & Altitude $h$ & 600 km\\
Terminals $K$ & 4 & Footprint radius $\rho_f$ & 200 km\\
Layers $L$ & 4 & Meta-atoms/layer $N$ & $10\times10$\\
Element spacing & $\lambda/2$ & Layer / feed spacing & $\lambda$ / $2\lambda$\\
Subcarriers $M$ & 1024 & Spacing $\Delta f$ / CP & 30 kHz / 72 samples\\
Control rate $f_u$ & 7.68 MHz & Slots per symbol $R$ & 274\\
Phase bits $B_\phi$ & 1--3, $\infty$ & Armijo $(c,\rho)$ & $(10^{-4},0.5)$\\
Random drops & 50 & Feed SNR $\mathrm{SNR}_0$ & 25 dB\\
\bottomrule
\end{tabular}
\end{table}
\begin{figure*}[t]
\centering
\subfloat[Convergence of Algorithm~\ref{alg:synthesis}\label{fig:convergence}]{\includegraphics[width=0.33\textwidth]{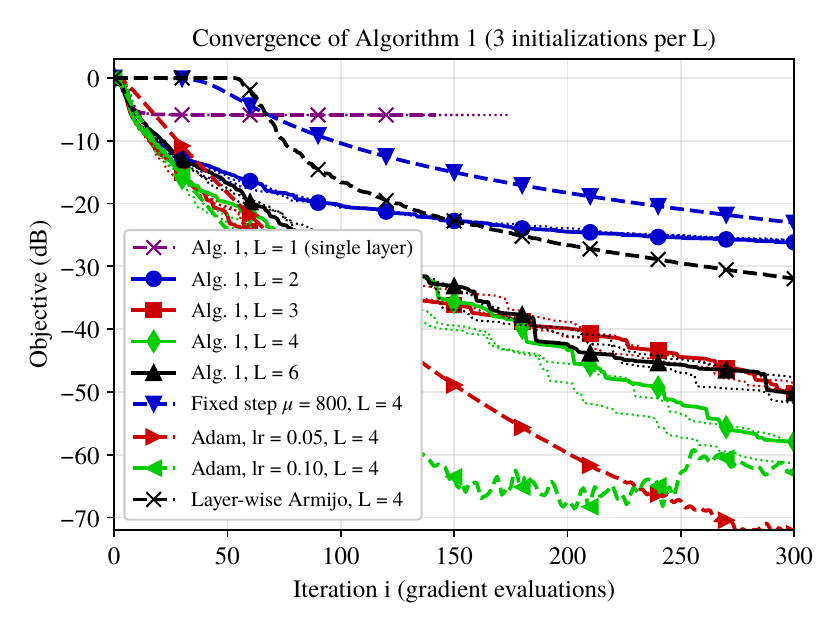}}\hfil
\subfloat[Synthesis error vs.\ phase resolution\label{fig:synthesis_error}]{\includegraphics[width=0.33\textwidth]{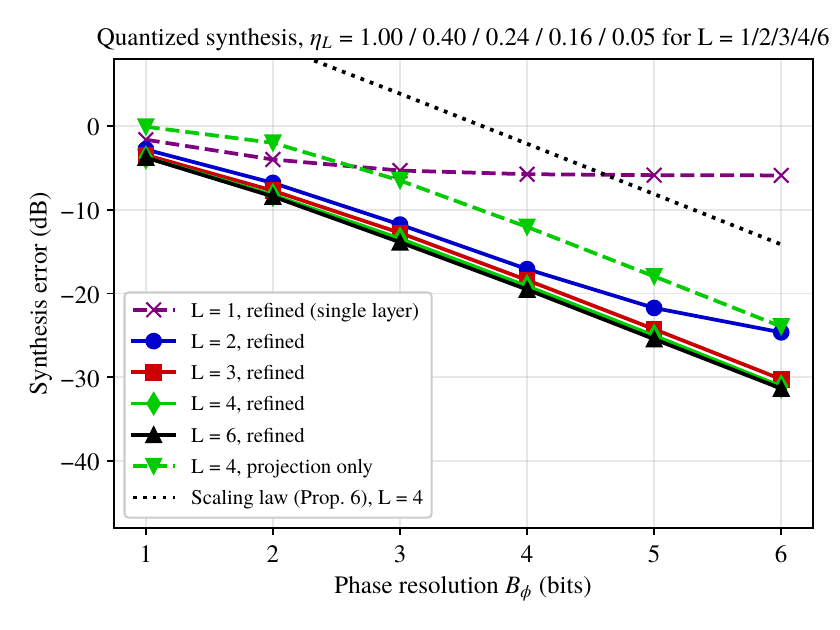}}\hfil
\subfloat[In-band SIR and out-of-band power vs.\ control rate\label{fig:control}]{\includegraphics[width=0.33\textwidth]{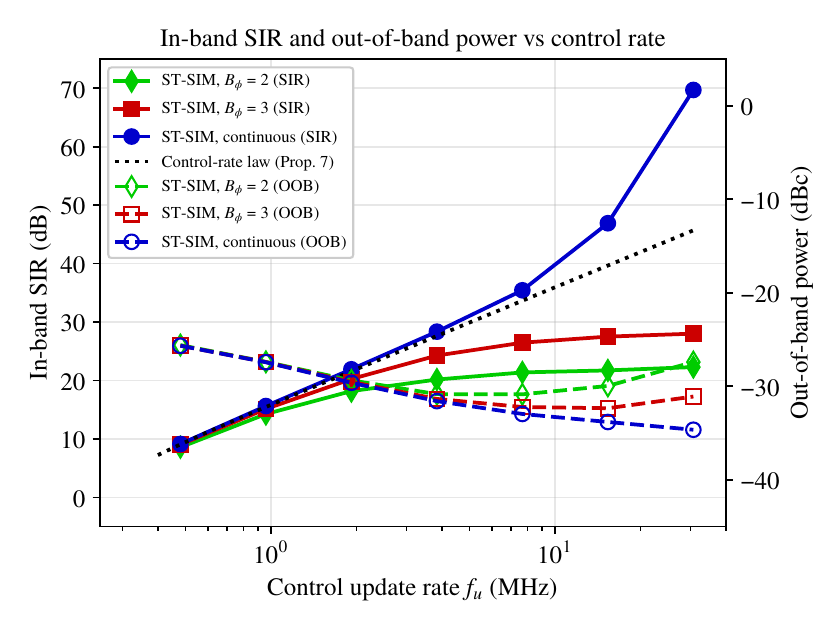}}
\caption{Realizability ($K=4$, $L=4$, $N=100$, $R=274$ unless stated). (a) Objective of Step~1 of Algorithm~\ref{alg:synthesis} for $L=1,2,3,4,6$ (three initializations each), fixed-step, Adam, and layer-wise descent for $L=4$. (b) Synthesis error with quantization-aware refinement ($L=1,2,3,4,6$) and projection only ($L=4$), with the scaling law of Proposition~\ref{thm:quant}. (c) In-band SIR (left) and out-of-band power (right) versus $f_u$, with the law of Proposition~\ref{prop:zoh}.}
\label{fig:hw}
\end{figure*}
Three benchmarks are considered. \emph{Benchmark~1} is a space-only SIM of the same geometry \cite{An2025,Lin2024} synthesizing the max--min multicast beam $\bm b^\star=\arg\max_{\norm{\bm b}_2=1}\min_k\abs{\bm a_k^H\bm b}^2$, computed by multi-start weighted power iteration, which is not certified to be globally optimal but is never worse than the sum $\bm\Psi^\star\one_K$ of the decoupled modes, with its feed shifted by $\nu_c^\star$; it represents every single-translation front end (Proposition~\ref{prop:common}). \emph{Benchmark~2}, ideal digital per-terminal precompensation, applies one digital translation per terminal and a lossless zero-forcing precoder, i.e., radiates $\bm\Psi^\star\bm d_D(u)$ exactly; it requires $K$ RF chains and a lossless $N\times K$ multibeam network and upper-bounds every exact realization of Theorem~\ref{thm:perfect}. \emph{Benchmark~3} is the payload of Benchmark~1 with every terminal correcting its residual CFO, including the error of Moose's two-symbol estimator, $\mathrm{var}(\hat\epsilon_k)=1/(4\pi^2M\,\mathrm{SNR}_k)$ \cite{Moose1994}, and, where stated, two pilot symbols per 14-symbol slot. The no-compensation reference is the space-only SIM without frequency correction. Finally, a single-layer time-modulated aperture is also investigated, $\bm g(u)=\bm\Phi_1(u)\bm f$, with the same feed, phase resolution, control rate, and synthesis objective as the ST-SIM (the one-layer case of Algorithm~\ref{alg:synthesis}), isolating the contribution of stacking.

The performance metrics are obtained as follows. The aperture field is evaluated on a four-times oversampled grid, Doppler-shifted, band-limited to $\abs f\le M\Delta f/2$ by an ideal receiver filter, and sampled; for terminal $k$, the response $s_k[m]=e^{\jmathc2\pi f_{D,k}mT_s}\bm a_k^H\widehat{\bm g}[m]$ over the DFT window gives the ICI matrix $\bm F\diag(\bm s_k)\bm F^H$, whose diagonal equals $\overline{s_k}$; hence $P_{d,k}=\abs{\overline{s_k}}^2$, and $P_{i,k}=\overline{\abs{s_k}^2}-P_{d,k}$ is the average interference for i.i.d.\ data, which includes switching images, quantization noise, and inter-mode leakage. Since the receiver filter acts on the product of the aperture response and the OFDM waveform, this cyclic model was verified by a full waveform simulation (random CP-OFDM symbols through the time-varying aperture, Doppler shift, receiver filter, decimation, and FFT): over the drops of Fig.~\ref{fig:footprint}, the per-terminal SIR of the ST-SIM differs by $0.3$~dB in median ($0.7$~dB at the 95th percentile), the cyclic model being slightly pessimistic, and the worst-terminal SINR by less than $0.1$~dB; over the control rates of Fig.~\ref{fig:control}, the average SIR of the full waveform differs from the cyclic model by $0.1$ to $0.7$~dB with continuous and three-bit phases, except by $-2.3$~dB with continuous phases at $15.36$~MHz, where the SIR exceeds $44$~dB. The SINR is $\gamma_k=\mathrm{SNR}_0P_{d,k}/(\mathrm{SNR}_0P_{i,k}+N)$ with $\mathrm{SNR}_0=PN/\sigma^2$, i.e., the SNR of a lossless full-aperture beam including the array gain $N$ ($5$~dB per meta-atom at $25$~dB for $N=100$), and, since all terminals decode the same broadcast stream, the figures report the worst-terminal SINR $\min_k\gamma_k$ and the broadcast spectral efficiency $\min_k\log_2(1+\gamma_k)$. All schemes are compared at \emph{equal long-run radiated peak energy}, defined by $\sup_u\norm{\bm g^\star(u)}_2^2=P_{\rm pk}(\bm\Psi)=1$ for every ideal field $\bm g^\star(u)=\bm\Psi\bm d_D(u)$, with $P_{\rm pk}$ the long-run peak energy \eqref{eq:peak}, and by $\widehat{\bm g}_r=\min\{1,\max_{r'}\norm{\alpha^\star\bm g_{r'}}_2^{-1}\}\alpha^\star\bm g_r$ for every realized field, i.e., the gain $\alpha^\star$ of \eqref{eq:alphastar} fits the realized field to the ideal target and is reduced if the energy exceeds one within the simulated window, as required by passivity (Lemma~\ref{lem:energy}); the \emph{equal feed power} figures of merit additionally scale the stack-based schemes by their simulated transmission efficiency $\eta_L=\overline{\norm{\bm G\bm f}_2^2}$.

\subsection{Realizability: Convergence, Quantization, and Control Rate}\label{sec:res_synth}
In Fig.~\ref{fig:convergence}, the objective of Algorithm~\ref{alg:synthesis} decreases monotonically (Proposition~\ref{prop:armijo}). The single layer stalls at $-5.9$~dB for every initialization: a phase-only layer cannot change the amplitude of any element excitation, whereas the multi-tone target \eqref{eq:dD} requires element amplitudes that vary in time, so that inter-layer mixing is necessary rather than merely helpful. For $L\ge3$ it reaches $-50$ to $-61$~dB after $300$ iterations and the three initializations of each depth end within $4$~dB; $L=2$ stalls near $-26$~dB, lacking the degrees of freedom to shape four frequency-translated modes. Fixed-step and layer-wise descent reach only $-23$ and $-32$~dB; Adam descends faster ($-68$ and $-73$~dB at learning rates $0.1$ and $0.05$) but is not monotone ($83$ and $36$ ascending steps out of $300$) and must be tuned per problem; since the quantization floor of $-13$~dB lies more than $40$~dB above the stationary points of either method, the monotone, tuning-free descent of Algorithm~\ref{alg:synthesis} is preferable for an autonomous payload.

Phase resolution does not change this picture (Fig.~\ref{fig:synthesis_error}): the single layer stays between $-5.3$ and $-5.9$~dB from three bits on, whereas two layers reach $-26$~dB with continuous phases. In Fig.~\ref{fig:synthesis_error}, the refinement of Step~3 lowers the four-layer error after projection from $-2.0$, $-6.5$, and $-12.0$~dB to $-8.2$, $-13.4$, and $-19.1$~dB for $2$, $3$, and $4$~bits, which is worth more than one bit, and all refined curves follow the $6$-dB-per-bit slope of Proposition~\ref{thm:quant} ($5.2$ and $5.7$~dB per bit from two to four bits for $L=4$). Beyond $L=3$ depth barely helps ($-12.8$, $-13.4$, and $-13.9$~dB at three bits for $L=3$, $4$, and $6$), whereas $\eta_L$ drops from $0.40$ to $0.24$, $0.16$, and $0.05$, so the simulated diffraction loss limits the useful depth. In Fig.~\ref{fig:control}, the in-band SIR with continuous phases follows Proposition~\ref{prop:zoh} within $1$~dB up to $f_u=3.84$~MHz and exceeds it from $7.68$~MHz on ($35$~dB on average, $31$~dB for the terminal with the largest tone), where the receiver filter rejects part of the images, whereas two and three bits saturate at $22$ and $28$~dB, the quantization floor. The gain at high rates is not free: the out-of-band fraction decreases from $-26$~dBc at $0.48$~MHz to $-33$~dBc at $7.68$~MHz for continuous phases, but rises again for quantized phases ($-27.5$~dBc for two bits at $30.72$~MHz), whose quantization noise spreads over a bandwidth proportional to $f_u$; generated in the radiating aperture, beyond any payload filter, these images must be controlled through $f_u$, $B_\phi$, and the control waveform.

\begin{figure*}[t]
\centering
\subfloat[Radiated power vs.\ direction and frequency\label{fig:sf_map}]{\includegraphics[width=0.33\textwidth]{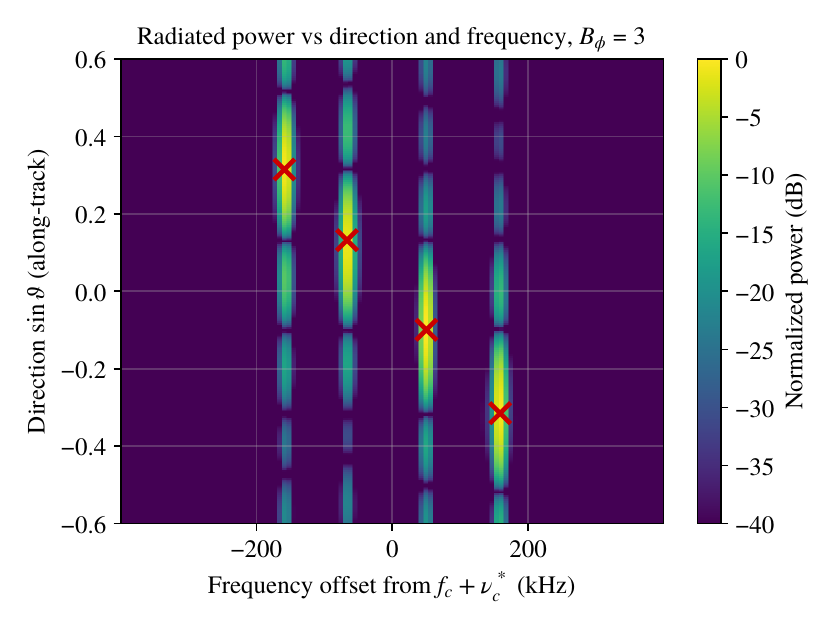}}\hfil
\subfloat[Worst-terminal SINR vs.\ footprint radius\label{fig:footprint}]{\includegraphics[width=0.33\textwidth]{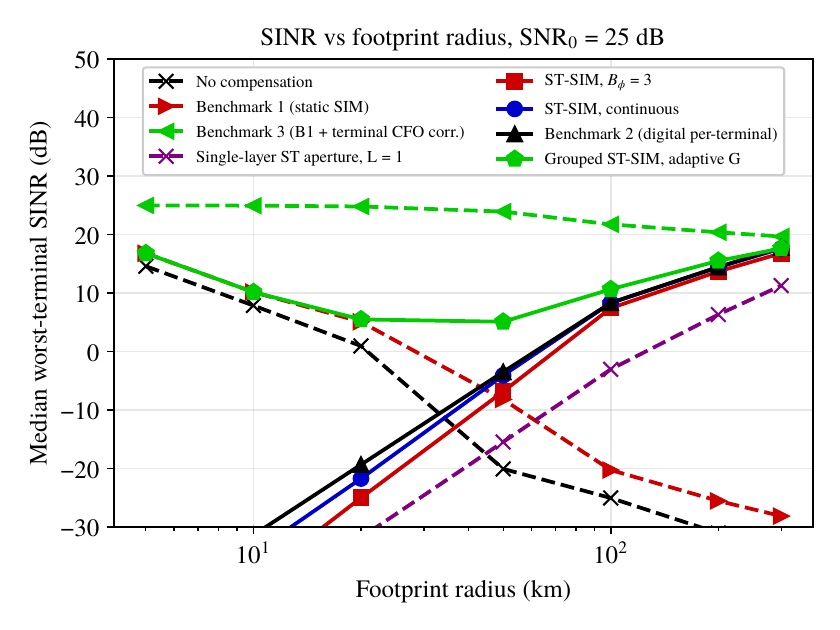}}\hfil
\subfloat[Worst-terminal SINR vs.\ aperture size\label{fig:aperture}]{\includegraphics[width=0.33\textwidth]{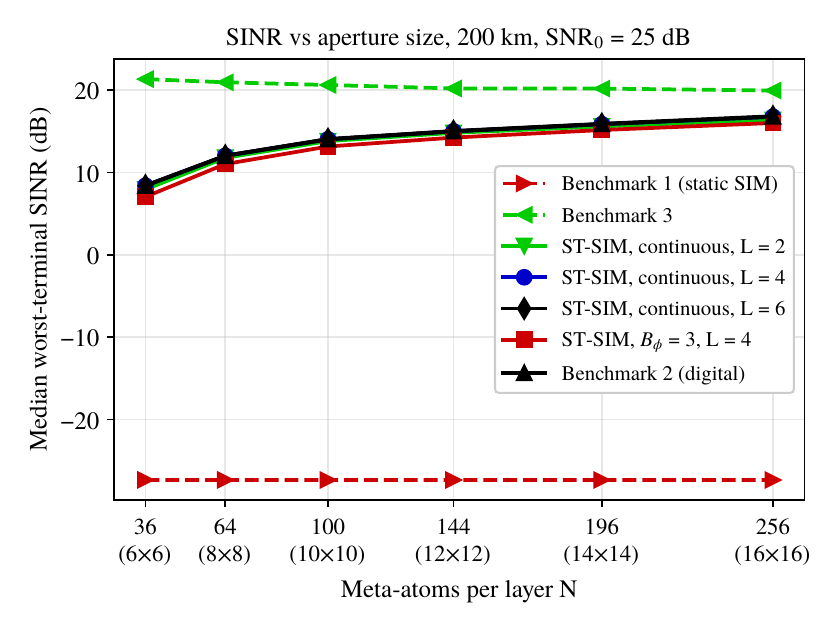}}\\[-1mm]
\subfloat[Broadcast spectral efficiency vs.\ $K$\label{fig:sumrate}]{\includegraphics[width=0.33\textwidth]{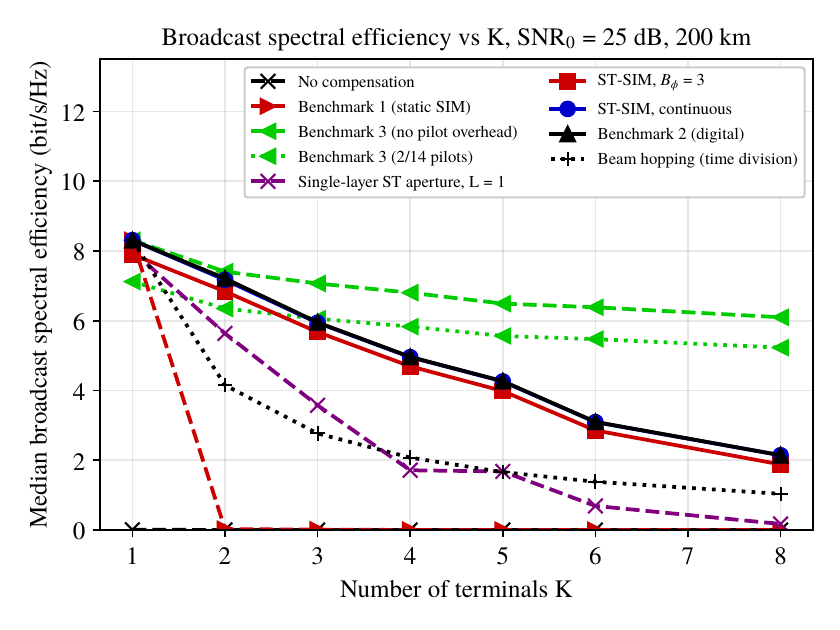}}\hfil
\subfloat[Grouped compensation, $G=1,\ldots,K$\label{fig:grouping}]{\includegraphics[width=0.33\textwidth]{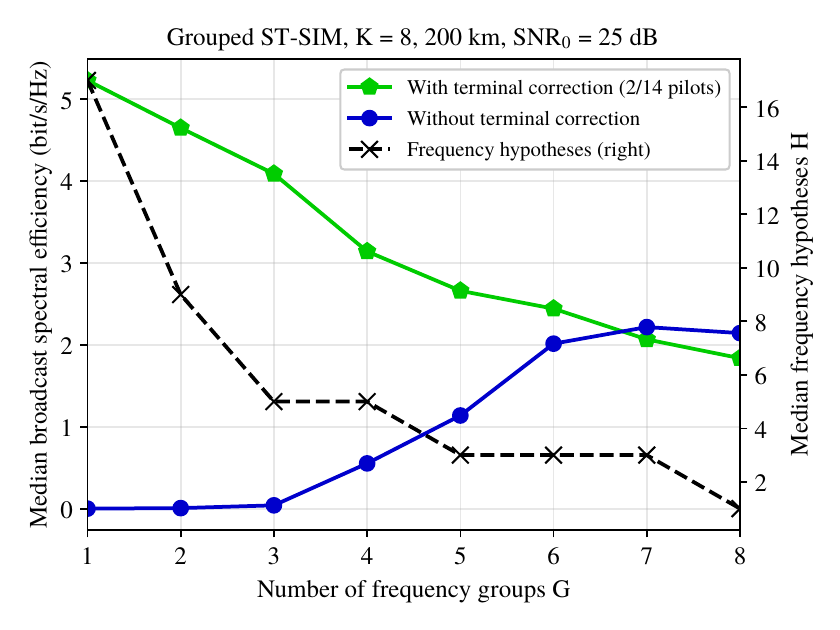}}\hfil
\subfloat[Robustness to position errors\label{fig:robust}]{\includegraphics[width=0.33\textwidth]{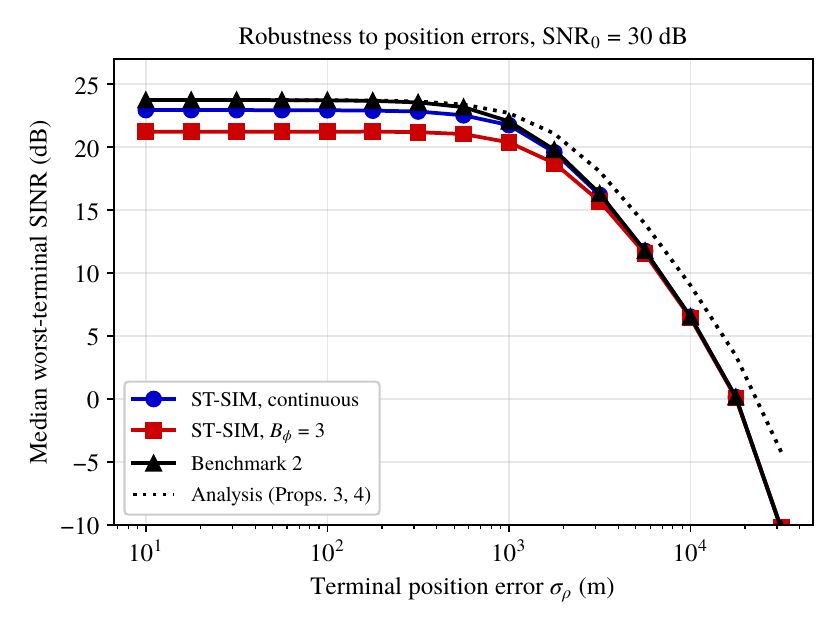}}
\caption{Link-level performance ($K=4$, $L=4$, $N=100$, $\mathrm{SNR}_0=25$~dB, $200$-km footprint, $50$ drops, equal long-run radiated peak energy, unless stated). (a) Radiated power of the three-bit stack versus along-track direction and frequency offset for four terminals on the ground track (crosses: terminal directions and tones). (b), (c) Median worst-terminal SINR versus footprint radius and versus aperture size ($L=2$, $4$, $6$; $30$ drops). (d) Median broadcast spectral efficiency versus $K$. (e) Grouped ST-SIM for $K=8$: spectral efficiency with and without terminal-side correction of the in-group residual, and frequency hypotheses. (f) Worst-terminal SINR at $\mathrm{SNR}_0=30$~dB versus terminal-position error, with the prediction of Propositions~\ref{prop:ici} and~\ref{prop:ephemeris}.}
\label{fig:link}
\end{figure*}

\subsection{What Time Modulation Adds to a SIM}\label{sec:res_link}
Fig.~\ref{fig:sf_map} shows the capability that distinguishes the ST-SIM from a static SIM: each terminal direction receives its own tone $-(f_{D,k}-\nu_c^\star)$ in a beam nulled toward the other terminals, i.e., $\bm H\approx\bm D_\beta$ in~\eqref{eq:coupling}; at each terminal, the other tones are at least $37$~dB down, so the $26$-dB in-band SIR of the three-bit stack is set by quantization and switching noise, not by inter-terminal leakage. Along the pass, Benchmark~1 attains an average SIR of only $-32$ to $+3$~dB, since the common-translation residual reaches $151$~kHz ($5.0\Delta f$) near zenith, whereas the ST-SIM keeps $35$--$49$~dB (continuous) and $21$--$27$~dB (three bits).

Fig.~\ref{fig:footprint} shows two regimes. In small footprints, the Doppler span is small, and the terminals are not spatially resolvable: Benchmark~1 performs well ($16.8$~dB at $5$~km), whereas the ST-SIM and Benchmark~2, which must separate the terminals to give them different frequencies, pay the penalty of Proposition~\ref{prop:solution} and Lemma~\ref{lem:peak} ($-42$~dB at $5$~km for the ideal field; the approximate realizations of the stack lie $2$--$6$~dB below it, trading leakage against energy). The single-layer aperture, whose interference floor lies near $13.5$~dB in the baseline geometry, stays $8.1$ and $6.3$~dB below the four-layer stack at $200$ and $300$~km, the gain of stacking at equal radiated peak energy, phase resolution, and control rate. The SINR-optimal translation $\nu_c^{\rm SINR}$ leaves Benchmark~1 unchanged up to $50$~km and raises it by $5.0$, $4.7$, and $3.2$~dB at $100$, $200$, and $300$~km, where it exploits the zeros of the ICI law, still more than $20$~dB below the ST-SIM. The curves cross at about $45$~km; beyond it, Benchmark~1 collapses, while the continuous-phase ST-SIM stays within $0.2$~dB of Benchmark~2 in median from $100$~km on ($14.4$~dB at $200$~km and $17.6$~dB at $300$~km, $0.7$--$0.8$~dB less with three bits). The long-run peak-to-average ratio of the ideal field, paid by every scheme, is $2.2$~dB in median at $200$~km and $4.6$~dB at $5$~km. Benchmark~3 attains $20$--$25$~dB at every radius, i.e., $6.0$~dB more than the ST-SIM at $200$~km, for two reasons inherent to per-terminal translation at the payload, digital or electromagnetic: its static field has no peak-to-average penalty and delivers a common gain to all terminals instead of separating them; it presumes, however, that every terminal has acquired and tracks a residual of several subcarrier spacings. At equal feed power, where every stack-based field is divided by the fitted gain $\abs{\alpha^\star}^2$ of \eqref{eq:alphastar} and Benchmark~2 is unchanged, the stack-based schemes lose about $8$~dB to the lossless Benchmark~2 ($\eta_L=0.16$), the simulated diffraction loss of the normalized cascade model: in the baseline geometry, Benchmark~2 attains $18.7$~dB, Benchmark~3 $10.4$~dB, and the ST-SIM $10.9$ and $11.3$~dB with continuous and three-bit phases, the latter reflecting a smaller fitted gain rather than a benefit of quantization (medians over the drops at $200$~km: $12.8$, $8.8$, and $14.5$~dB for Benchmark~3, the ST-SIM, and Benchmark~2). The lossless single layer ($\eta_1=1$) reaches $12.0$~dB at equal feed power in the baseline geometry, and the four-layer stack overtakes it from $\mathrm{SNR}_0\approx27$~dB on, where the interference floor of the single layer dominates. The optimal depth at equal feed power is smaller: in the baseline geometry, the equal-feed worst-terminal SINR at $\mathrm{SNR}_0=15/25/40$~dB is $6.6/12.0/13.4$~dB for $L=1$, $5.0/14.8/27.1$~dB for $L=2$, $2.8/12.8/26.0$~dB for $L=3$, $0.9/10.9/24.7$~dB for $L=4$, and $-4.1/5.8/20.4$~dB for $L=6$, whereas $L=2$ to $6$ coincide within $0.4$~dB at equal radiated energy; two layers, whose $-26$-dB synthesis error lies below the noise at these SNRs, overtake the single layer from $\mathrm{SNR}_0\approx20$~dB and are the best equal-feed choice in the normalized model, where each additional layer costs about $4$~dB of diffraction loss. The near-digital performance thus holds at equal radiated peak energy only, and the single-layer comparison is conditional on the SNR, on the geometry, and on the diffraction loss of the model. The baseline $\mathrm{SNR}_0=25$~dB is optimistic for current high-power amplifiers and Ka-band link budgets \cite{3GPP38821}; at $\mathrm{SNR}_0=15$~dB, the median worst-terminal SINR at $200$~km is $4.5$~dB for the ST-SIM and Benchmark~2 and $10.4$~dB for Benchmark~3, and the crossover with Benchmark~1 moves to about $59$~km, which calls for the larger apertures of Section~\ref{sec:res_scaling}.

\subsection{Aperture Size and Number of Terminals}\label{sec:res_scaling}
Fig.~\ref{fig:aperture} sweeps the aperture from $6\times6$ to $16\times16$ meta-atoms at $200$~km. The static Benchmark~1 is limited by the Doppler floor at $-27$~dB and does not improve with $N$, and Benchmark~3 stays at $20$--$21$~dB, whereas the ST-SIM gains $8.4$~dB from $N=36$ to $256$, because a larger aperture resolves the terminals with less modal energy (Proposition~\ref{prop:solution}); with $L=4$ the continuous-phase stack tracks Benchmark~2 within $0.1$~dB for every $N$, three bits cost $0.8$--$1.4$~dB, $L=2$ falls $0.2$--$0.5$~dB short, and $L=6$ brings no further gain. 

The same mechanism governs Fig.~\ref{fig:sumrate}: the broadcast spectral efficiency of the continuous-phase ST-SIM tracks Benchmark~2 within $0.05$~bit/s/Hz for every $K$ with a single RF chain, three bits cost at most $0.4$~bit/s/Hz, and Benchmark~1 collapses for $K\ge2$, and the single-layer aperture loses $3.2$~bit/s/Hz to the four-layer stack at $K=4$. Benchmark~3 without pilot overhead reaches $6.8$~bit/s/Hz at $K=4$ and exceeds the ST-SIM ($5.0$~bit/s/Hz) for $K\ge2$, since a static beam does not pay the separation cost that grows with $K$; with two pilot symbols per slot it drops to $5.8$~bit/s/Hz, and the ST-SIM leads unless CFO tracking costs less than $3\%$, $16\%$, and $27\%$ of the resources for $K=2$, $3$, and $4$. Time-division beam hopping serves each terminal only during its dwell: $\log_2(1+\mathrm{SNR}_0)/K=2.1$~bit/s/Hz at $K=4$.

\subsection{Grouped Compensation and Robustness}\label{sec:res_group}
Remark~\ref{rem:grouped} is evaluated with contiguous partitions of the Doppler-ordered terminals, each group receiving the midpoint of its Doppler range. In Fig.~\ref{fig:footprint}, \emph{adaptive grouping} selects for every drop the partition that maximizes the spectral efficiency predicted from the ideal grouped field, which the stack then synthesizes without terminal-side correction. It follows Benchmark~1 in small footprints ({$G=1$ in all drops up to $10$~km and in $96\%$ at $20$~km}) and the per-terminal ST-SIM in large ones ($G=K$ in {$88\%$} of the drops at $300$~km), and exceeds both in between: at $50$~km, where it mostly selects $G=2$ or $3$, its median SINR of {$5.1$}~dB lies {$9.2$}~dB above the per-terminal ST-SIM, and at $100$ and $200$~km the gain is still {$2.3$ and $1.1$}~dB, so the mode switch at the crossover becomes a smooth transition. Fig.~\ref{fig:grouping} sweeps $G$ for $K=8$ terminals in a $200$-km footprint with the partition that minimizes the largest in-group residual. Without terminal-side correction, the spectral efficiency grows from $0.0$~bit/s/Hz at $G=1$ (Benchmark~1) to $2.1$~bit/s/Hz at $G=K$, since in-group residuals of several subcarrier spacings are fatal; with correction of the in-group residual, it decreases from {$5.2$}~bit/s/Hz at $G=1$ (Benchmark~3), because separating groups costs energy, while the frequency search of the terminal shrinks from {$17$ to $9$} and $5$ hypotheses on a $\Delta f/2$ grid for $G=2$ and $3$ at {$4.7$ and $4.1$}~bit/s/Hz. For $K=4$, $G=2$ reduces the search from $13$ to $5$ hypotheses at a {$6\%$} loss in spectral efficiency. Grouping thus turns the choice between payload- and terminal-side correction into a continuum.

\label{sec:res_robust}The SINR loss of the ST-SIM is about $1$~dB for terminal-position errors up to $1$~km, as for Benchmark~2 and as predicted by Propositions~\ref{prop:ici} and~\ref{prop:ephemeris} (Fig.~\ref{fig:robust}), and $0.3$~dB for satellite-velocity errors up to $10$~m/s. {The initial acquisition of an unregistered terminal is outside the scope of the ST-SIM, since no dedicated mode exists before its position is known (Section~\ref{sec:model}); for the reacquisition of a registered terminal and for tracking during a session,} the ST-SIM delivers a carrier-Doppler-free signal{, so that} a terminal needs a single frequency hypothesis, whereas a common-translation payload leaves $9$ ($13$) hypotheses in median (95th percentile) at $200$~km and $11$ ($17$) at $300$~km; the oscillator offset adds the same range to all schemes.

\section{Conclusion}\label{sec:conclusion}
This paper introduced space--time SIMs, which pair each spatial mode of a multilayer stack with its own frequency translation, and applied them to Doppler-resilient LEO broadcasting with a single RF chain. For the ideal aperture field, stacking convolves the temporal spectra of the layers (Lemma~\ref{lem:conv}), the minimum-energy space--frequency synthesis has a closed-form solution governed by the terminal geometry (Proposition~\ref{prop:solution}), full column rank of the steering matrix is necessary and sufficient for decoupling $K$ terminal-specific translations (Theorem~\ref{thm:perfect}), grouping relaxes this condition (Remark~\ref{rem:grouped}), and any common translation, i.e., any static SIM or single-translation time-modulated surface, leaves half the Doppler span (Proposition~\ref{prop:common}). Algorithm~\ref{alg:synthesis} decreases the synthesis objective monotonically to stationary points, and explicit laws quantify quantization, control rate, and ephemeris errors. Under the ideal phase-only wave-domain model, a four-layer stack reached synthesis errors of $-58$~dB with continuous and $-13$~dB with three-bit phases, removed the Doppler floor of a static SIM above a footprint radius of about {$45$}~km, matched ideal digital per-terminal precompensation within $0.5$~dB and $0.05$~bit/s/Hz at equal radiated peak energy, and kept an in-band SIR of $26$~dB with three bits at a $7.68$-MHz control rate{, whereas a single time-modulated layer stalled at a synthesis error of $-6$~dB, which identifies inter-layer mixing as the enabling element of the architecture}.

The results also delimit the regime of the architecture. Separating terminals costs modal energy that grows with their number and shrinks with the aperture ($+${$8.4$}~dB from $36$ to $256$ meta-atoms per layer). A static beam with terminal-side CFO correction attained a higher SINR and, {with dedicated pilots, a higher broadcast spectral efficiency for three or more terminals}; the ST-SIM led for two terminals when CFO tracking costs more than {$3\%$} of the resources, and always delivered a carrier-Doppler-free signal that removes the frequency search of {registered terminals}, with grouped compensation providing a continuum between the two approaches. At equal feed power, each layer of the normalized cascade model costs about $4$~dB of diffraction loss, which favors two or three layers. 


\appendices
\section{Light-Time Doppler Model}\label{app:lighttime}
Let $\mathcal T_k(u)\triangleq u+\tau_k(u)$, where $\tau_k$ solves $c\,\tau_k(u)=\norm{\bm r_k(u+\tau_k(u))-\bm r_s(u)}_2\triangleq d_k(u)$ \cite{montenbruck2002satellite}, and let $\bm u_k$ be the unit vector from $\bm r_s(u)$ to $\bm r_k(u+\tau_k)$. Since $\abs{\bm u_k^T\dot{\bm r}_k}<c$ and $\abs{\bm u_k^T\dot{\bm r}_s}<c$, the light-time equation has a unique solution, $\mathcal T_k$ is strictly increasing, and the retarded transmit time $u_k(t)=\mathcal T_k^{-1}(t)$ is well defined.
\begin{proposition}[Light-time LEO Doppler]\label{prop:doppler}
The instantaneous carrier Doppler frequency of \eqref{eq:exactrx} is
\begin{equation}
f_{D,k}^{\rm ex}(t)=-f_c\frac{\dot\tau_k(u)}{1+\dot\tau_k(u)}
=-f_c\frac{\bm u_k^T\bm v_k^{\rm rel}}{c-\bm u_k^T\dot{\bm r}_s},\quad u=u_k(t),
\label{eq:exactdoppler}
\end{equation}
so that $f_{D,k}^{\rm ex}(t)=-\frac{f_c}{c}\bm u_k^T\bm v_k^{\rm rel}\big[1+\mathcal O(\norm{\dot{\bm r}_s}_2/c)\big]$.
\end{proposition}
\begin{IEEEproof}
The carrier phase $\varphi_k(t)=-2\pi f_c\tau_k(u_k(t))$ gives $f_{D,k}^{\rm ex}=\frac{1}{2\pi}\frac{d\varphi_k}{dt}=-f_c\dot\tau_k\frac{du}{dt}$ with $dt/du=1+\dot\tau_k$, i.e., the first form. Differentiating the light-time equation, with the terminal at $u+\tau_k(u)$, gives $c\dot\tau_k=\bm u_k^T[(1+\dot\tau_k)\dot{\bm r}_k-\dot{\bm r}_s]$, hence $1+\dot\tau_k=(c-\bm u_k^T\dot{\bm r}_s)/(c-\bm u_k^T\dot{\bm r}_k)$ and $\dot\tau_k/(1+\dot\tau_k)=\bm u_k^T\bm v_k^{\rm rel}/(c-\bm u_k^T\dot{\bm r}_s)$, the second form; expanding $1/(c-\bm u_k^T\dot{\bm r}_s)$ in $\norm{\dot{\bm r}_s}_2/c$ completes the proof.
\end{IEEEproof}
For LEO motion the relative correction is at most $\norm{\dot{\bm r}_s}_2/c=2.5\times10^{-5}$, so the first-order form is used throughout.

\balance
\bibliographystyle{IEEEtran}
\bibliography{main}

\end{document}